\documentclass[10pt,reqno]{amsart}
\usepackage{hyperref}
\usepackage{latexsym,amsmath}
\usepackage{times}
\usepackage{amsfonts}
\usepackage{amssymb}
\usepackage{latexsym}
\usepackage{fixmath}
\usepackage{bbm}
\usepackage{enumerate}
\usepackage{tikz}

\usepackage[utf8]{inputenc}
\usepackage[dvips]{epsfig}
\graphicspath{{./Fig_eps/}{./Eps/}}
\DeclareGraphicsExtensions{.ps,.eps}
\usepackage{xcolor}
\usepackage{fullpage}

\def\vec#1{\mathchoice{\mbox{\boldmath$\displaystyle#1$}}
{\mbox{\boldmath$\textstyle#1$}}
{\mbox{\boldmath$\scriptstyle#1$}}
{\mbox{\boldmath$\scriptscriptstyle#1$}}}

\newcommand\dd{{\mathrm d}}

\newcommand\G{\vec G}

\DeclareMathOperator{\Erw}{\mathbb E}
\DeclareMathOperator{\pr}{\mathbb P}
\newcommand\SIGMA{\vec\sigma}

\newtheorem{definition}{Definition}[section]
\newtheorem{claim}[definition]{Claim}

\newtheorem{theorem}[definition]{Theorem}
\newtheorem{lemma}[definition]{Lemma}
\newtheorem{proposition}[definition]{Proposition}

\newtheorem{fact}[definition]{Fact}

\newcommand\cA{\mathcal{A}}
\newcommand\cB{\mathcal{B}}
\newcommand\cC{\mathcal{C}}

\newcommand\cG{\mathcal{G}}

\newcommand\cU{\mathcal{U}}

\newcommand\cH{\mathcal{H}}

\newcommand\cL{\mathcal{L}}

\newcommand\cP{\mathcal{P}}

\newcommand\cY{\mathcal{Y}}

\newcommand\eps{\varepsilon}
\newcommand\ZZ{\mathbf{Z}}

\newcommand{\N}{\mathbb{N}}
\newcommand\FF{\mathbb{F}}

\newcommand{\vecone}{\vec{1}}

\newcommand{\Po}{{\rm Po}}

\newcommand{\Be}{{\rm Be}}

\newcommand\RR{\mathbb{R}}

\newcommand{\whp}{w.h.p.}

\newcommand{\E}{\mathbb{E}}

\renewcommand{\P}{\mathcal{P}}

\newcommand{\rG}{{\vec G}}

\newcommand{\rDelta}{\vec \delta}

\newcommand{\sm}{s_{\max}}
\newcommand{\PSI}{\vec\psi}
\newcommand{\MU}{\vec\mu}
\newcommand{\PP}{\cP^2_*(\Omega)}
\newcommand{\Erdos}{Erd\H{o}s}
\newcommand{\Renyi}{R\'enyi}
\begin{document}

\title{The Mutual information of LDGM codes}

\author{Jan van den Brand}
\author{Nor Jaafari}
\thanks{This work has been submitted to the IEEE for possible publication. Copyright may be transferred without notice, 
after which this version may no longer be accessible.}
\address{Jan van den Brand, {\tt janvdb@kth.se}, KTH Royal Institute of Technology, Stockholm, Sweden.}
\address{Nor Jaafari, {\tt jaafari@math.uni-frankfurt.de}, Goethe University, Mathematics Institute, Frankfurt, Germany.}
\begin{abstract}
We provide matching upper and lower bounds on the mutual information in noisy reconstruction of parity check codes and thereby prove a long-standing conjecture by Montanari [IEEE Transactions on Information Theory 2005]. Besides extending a prior concentration result of Abbe and Montanari [Theory of Computing 2015] to the case of odd check degrees, we precisely determine the conjectured formula for code ensembles of arbitrary degree distribution, thus capturing a broad class of capacity approaching codes.
\end{abstract}

\maketitle

\section{Introduction}
Sparse random binary matrices provide a natural way of encoding messages without exhausting the transmission rate. Let $m$ be a number larger than the blocklength of a message $\xi\in\FF_2^n$. By choosing a random generator matrix $\vec A$ over the field $\FF_2^{m\times n}$,  we obtain a codeword $\vec x$ by simple matrix multiplication $\vec x = \vec A \xi$. Given $\vec A$ and a noisy observation $\vec {\tilde x}$ obtained  from the \emph{binary memoryless symmetric} (BMS) channel, we can most likely recover $\xi$ by solving the system of linear equations. That is, provided $m$ is sufficiently large and the matrix $\vec A$ imposes adequate redundancy averaged over the bits of $\xi$. Properly structured \emph{sparse} random matrices $\vec A$ induce the class of \emph{Low Density Generator Matrix} (LDGM) codes. LDGM codes have been known for many decades, and although we can describe them in very few lines, to date the status of research has hardly advanced, partly due to the fact that while simplifying the encoding and decoding process, sparsification also severely exacerbates the analysis. This is unsurprising but remarkable, as we know that at least for some code constructions these codes perform remarkably well. This paper aims to break the first ground by proving a precise formula for the mutual information of LDGM codes which was previously conjectured by Montanari \cite{montanari}.

\subsection{Optimal codes on graphs}
As may be expected, structured  sparse code ensembles are readily constructed from bipartite graphs. Also known as \emph{tanner graphs} (\emph{factor graphs}), they consist of \emph{variable nodes} representing bits of a signal on the left-hand side and \emph{check nodes} (\emph{factor nodes}) representing parity check equations on the right-hand side. In a more general context factor graphs are used to model \emph{constraint satisfaction problems} (CSPs) by having factor nodes impose constraints on participating variable nodes. In LDGM codes each variable node participates in at least one parity equation via its neighboring check nodes: If a check node constrained to $x_j \in \FF_2$  is adjacent to $k$ bits $\xi_1,\ldots,\xi_k$, then codewords need to satisfy $\xi_1 + \ldots + \xi_k =  x_j$ in $\FF_2$. The intersection of codewords satisfying all parity check equations modeled by the graph form a linear code. 

The crux is that while local interactions are constructed in a simple fashion, the global structure of the code has to perform in a complex interplay, as to ensure efficient coding and decoding. Importantly, the placement and number of check nodes can be chosen in a sophisticated way as to keep encoding and decoding complexity low, while maintaining a sufficient amount of redundancy. In a nutshell these so-called \emph{standard code ensembles} achieve their performance by imposing a specific degree distribution on the nodes of the factor graph. For example, we constrain check nodes to perform a parity check on a constant number $k$ of variable nodes with degree distribution $D$. The mutual information of $(D,k)$-code ensembles has been rigorously studied in previous work. In a pioneering paper \cite{montanari} Montanari derives an upper bound on the mutual information for even $k$ and conjectures the bound to be tight. In subsequent work \cite{CRF} Abbe and Montanari were able to prove the existence of a limit for the mutual information in the same scenario. Our result comprehensively determines this limit for the mutual information of $(D,k)$-code ensembles for both even and odd $k$.

\subsection{The mutual information} In the last decade there has been a critical endeavor to analyze such \emph{standard code ensembles} with respect to error-free decodeability. The mutual information $I(X,Y)$ captures how much information the output of the channel $Y$ contains about the input of the channel $X$ and thus provides the essential measure to quantify information-theoretic limits of decodeability. As it artlessly entails bounds on error probabilities the mutual information proves to be key in many related areas of decoding noisy signals \cite{BDMK,BDMKLZ,BM,LelargeMiolane}, most importantly in the analyses of random linear codes \cite{thebook}. In the setting of BMS channels with noise $\eta\in(0,1/2)$ each bit is independently correctly transmitted with probability $1-\eta$ and flipped with probability $\eta$. Computing the exact mutual information in this configuration is a highly non-trivial task and despite a substantial amount of research \cite{GMU, KM1, KM2, montanari}, one is usually merely able to provide bounds and occasionally tailor formulas to individual scenarios. In a general case analysis for LDGM codes with given variable degrees Montanari \cite{montanari} derives an upper bound on the mutual information, subject to the condition that check degrees satisfy a convexity assumption. He conjectures the bound to be sharp, as it matches explicit formulas from auspicious but non-rigorous calculations in the spin-glass theory of statistical mechanics.

In this paper, building upon the indispensable groundwork by Montanari, we establish the Aizenman-Sims-Starr \cite{ASS} cavity computation for standard code ensembles to derive a matching lower bound and with that prove the conjectured bound to be tight. Furthermore, we introduce a new technique to drop the assumptions on the variable degree distribution and extend the results a comprehensive class of standard code ensembles.

\subsection{Results}
The following theorem proves Montanaris long-standing conjecture from \cite{montanari}. Particularly, it is the first result to precisely determine the mutual information in random LDGM codes with given variable degrees without imposing restrictions on the degree distribution or the magnitude of noise.
\medskip

The predicted formula comes in the form of a stochastic fixed-point equation. To state our Theorem denote by $\cP_0([-1,1])$ the set of probability distributions on $[-1,1]$ with mean zero and let $\pi\in\cP_0([-1,1])$. Fix a degree distribution $D$, a number $k>0$ and let $(\vec\theta_{i,j})_{i\ge 0,1\le j\le k}$ be independently identically distributed samples from $\pi$, and $\vec \gamma$ be chosen from $D$. Write $\Be(\eta)$ for the outcome of a random Bernoulli experiment with parameter $\eta$. Further, let $\Lambda(x)=x\ln x$ for $x\in(0,\infty)$ and $\Lambda(0)=0$. With a sequence $\vec J_1,\vec J_2,\ldots$ of independent copies of $\vec J=1-2\Be(\eta)$, we let
\begin{align*}
\textstyle
\cL(k,D,\eta)=\frac{1}{2}
\Erw \left[ \Lambda\left(
\sum_{\sigma \in \{\pm1\}} \prod_{a=1}^{\vec\gamma}
\left(1 +  \vec J_a \sigma \prod_{j = 1}^{k-1} \vec\theta_{a,j}\right)
\right)\right]
- \frac{k-1}{k}\Erw[\vec\gamma]
\Erw\left[\Lambda \left(
1 + \mathbf{J} \prod_{j=1}^k \pmb{\theta}_{0,j}
\right)
\right].
\end{align*}

\begin{theorem}
\label{thm:Codes}
Let $\cC_n$ be a random LDGM-code with blocklength $n$, variable degree distribution $D$ and check degree $k$. Let $\vec X$ be a message chosen uniformly at random from $\cC_n$. If $\vec Y$ is the message obtained by passing $\vec X$ through a memoryless binary symmetric channel with error probability $\eta>0$ then
\begin{align*}
\lim_{n\to\infty}\frac1n I(\vec X, \vec Y) = -\sup_{\cP_0([-1,1])}\cL(k,D,\eta) +\frac{\Erw[\vec\gamma]}{2k}\left((1 - \eta) \ln (1 - \eta) + \ln 2 + \eta \ln (\eta)\right) + \ln 2.
\end{align*}
\end{theorem}
The argument we develop to prove the mutual information in standard graph ensembles is reasonably general. We expect it to extend to LDPC code ensembles and similar problems of conditional random fields. Possibly our approach may facilitate simpler proofs to the analyses of spatially coupled codes. It is important to mention that while spatial coupling was invented to engineer error-correcting codes, the technique is now applied beyond the context of coding theory \cite{AHMU,HMU}. 

In Section \ref{sec:proofstrat}, we will therefore actually prove a general version of the main theorem to include a broader class of factor graph ensembles encompassing many problems related to conditional random fields. Theorem \ref{thm:Codes} deals with its most natural application being the problem of noisy reconstruction in standard LDGM code ensembles.

\subsection{Background and related work}
Following Shannon's work, early codes based on algebraic constructions were analyzed, only to realize that these codes do not saturate the capacity limit by a far margin. Progress stalled in the subsequent decades until the early 90s. The introduction of Turbo codes \cite{BGT} reignited interest, as Turbo codes were able to deliver performance close to the Shannon limit. Ensuing generalizations uncovered the power of codes based on graphs as Low Density Parity Check (LDPC) codes, originally put forward in Robert Gallagher's PhD thesis in 1962, were rediscovered. Having been neglected for quite some time,  
LDPC codes reemerged to a broad audience, when a series of papers by Luby, Mitzenmacher, Shokrollahi, Spielman and Steman \cite{LMSS1,LMSS2,LMSS3,LMSSS} followed up by work of Richardson, Shokrollahi and Urbanke \cite{RSU, RU} proved that Gallagher's parity check codes perform at rates close to Shannon capacity and can be designed such that efficient decoding is possible. 
Crucially, in constructing codes that approach capacity one has to keep in mind the performance of the decoder. Recently, parity check codes have become widely used and successfully implemented in the context of satellite communication, WiFi transmission and data protocols. At the same time significant progress has been made with regard to the design of optimal codes, that not only achieve capacity but also allow for efficient coding and decoding. This is usually achieved by prescribing a degree distribution on the variable nodes. 
Driven by the success of spatial coupling, the analysis of standard code ensembles has gathered tremendous momentum. Nonetheless, the construction and analysis of spatially coupled codes in particular is notoriously complicated.

In the analysis of linear codes, Sourlas' work \cite{Sourlas1, Sourlas2, Sourlas3} dating back to the early 90s provided a crucial link to the physics theory of spin systems that would betoken the path to numerous fruitful results. Many methods therefore base on developments in the rigorous theory of mean field spin glasses. So far the most promising analyses of standard LDGM code ensembles utilize Guerra and Toninelli's interpolation method \cite{GT} to provide general bounds on the inference threshold for graphical models. Subject to the condition that the generating function of the left degree distribution is convex, Abbe and Montanari \cite{CRF} show that the entropy of the transmitted message conditional to the received one concentrates around a well defined deterministic limit. In a previous work by Montanari \cite{montanari} the interpolation method was employed to lower bound the entropy by an asymptotic expression derived in \cite{FranzLeoneMontanariRicci}, which was obtained using heuristic statistical mechanics calculations. In a recent paper by Coja-Oghlan, Krkazala, Perkins and Zdeborov\'a \cite{CKPZ} the entropy was derived for \Erdos-\Renyi~ type LDGM graph ensembles. However, these are merely fragmentarily applicable for error-correcting codes, because the corresponding adjacency matrices are unstructured and may exhibit empty columns. 

 
\section{Proof Outline}
In this section we outline the pillars of our proof. To derive the mutual information in noisy observations for a broad class of random factor graph models, we make slightly more general assumptions than needed for codes. As such our result includes a number of problems related to conditional random fields. The \emph{teacher-student model} can be viewed as a natural generalization of the retrieval problem in BMS-distorted LDGM codes. Instead of $\{\pm 1\}$ we consider an arbitrary finite set $\Omega$ of possible bit values and generalize the parity check constraints to a finite set $\Psi$ of positive weight functions $\psi:\Omega^k\to(0,2)$ for some fixed $k>1$. At the basis of $\Psi$ is a probability space $(\Psi, p)$ with a prior distribution $p$ on the weight functions. We write $\PSI$ for a random choice from $p$. Further, we specify unweighted factor graphs $G=(V,F,(\partial a)_{a\in F})$ by their bipartition $V, F$ into variable nodes $V$ and check nodes $F$ as well as their neighborhood structure $(\partial a)_{a\in F}$, $\partial a\in V^k$. A (weighted) factor graph additionally carries weight functions $\psi_a:\Omega^{\partial a}\to(0,2)$ on each factor node $a\in F$ that locally evaluate signals (\emph{assignments}) $\sigma\in\Omega^V$ on the variable nodes of the graph $G=(V,F,(\partial a)_{a\in F},(\psi_a)_{a\in F})$. For a set of variable nodes of size $n$ we write $V=V_n$ and omit the index if it is apparent from the context. Commonly, if $\ell$ is an integer, we write $[\ell]=\{1,\ldots,\ell\}$ and identify $V_n$ with the set $[n]$. If $a$ is a check node we denote by $\partial a$ its neighborhood in $V^k$ and by $(\partial_1 a,\ldots,\partial_k a)$ the vector of its $k$ neighbors in ascending order. To exert supplementary methods and results from \cite{CKPZ}, throughout the paper we require $p$ to satisfy two assumptions \textbf{SYM} and \textbf{POS} that are on the one hand easily verified for the class of LDGM codes and a number of related applications and on the other hand directly imply the assumptions from \cite{CKPZ}. To this end let $\cP(\Omega)$ denote the set of probability distributions on $\Omega$. Further, let $\cP^2_*(\Omega)$ denote the set of all probability measures on $\cP(\Omega)$ whose mean corresponds to the uniform distribution.
\begin{description}
\item[SYM] Let $\xi=|\Omega|^{-k}\sum_{\tau\in\Omega^k}\Erw[\PSI(\tau)]$. For all $\sigma\in\Omega^k$ we have $\Erw[\PSI(\sigma)]=\xi$.
\item[POS] For all $\pi,\pi'\in\cP_*^2(\Omega)$ the following is true. With $\MU_1,\MU_2,\ldots$ chosen from $\pi$,
			$\MU_1',\MU_2',\ldots$ chosen from $\pi'$ and $\PSI\in\Psi$ chosen from $p$, all mutually independent,  we have
	\begin{align*}
	&\textstyle\Erw\left[\Lambda\left(\sum_{\tau\in\Omega^k}\PSI(\tau)\prod_{i=1}^ k\MU_i(\tau_i)\right)+(k-1)\Lambda\left(\sum_{\tau\in\Omega^k}\PSI(\tau)\prod_{i=1}^k \MU_i'(\tau_i)\right)\right.\\
	&\textstyle\qquad\qquad\qquad\qquad\qquad\qquad\qquad\qquad\quad\left.-k\Lambda\left(\sum_{\tau\in\Omega^k}\PSI(\tau)\MU_1(\tau_1)\prod_{i=2}^k\MU_i'(\tau_i)\right)\right]\ge 0.
	\end{align*}
\end{description}
Linear codes satisfy both \textbf{SYM} and \textbf{POS} as we will see in Section \ref{sec:proofstrat}. The general story now goes as follows. The teacher chooses a ground truth $\SIGMA^*\in\Omega^V$ uniformly at random that he finds himself unable to directly convey to his students. He uses the teacher-student model in which he may set up a random $(D,k)$-graph $\G^*=\G^*(n,p,\SIGMA^*)$ on $n$ variable nodes and $M=M(D,k,n)$ check nodes, where parity checks are chosen proportionally to the local evaluation, that is 
\begin{align*}
\pr[\psi_a = \psi]\propto p(\psi)\psi(\SIGMA^*(x_1),\ldots,\SIGMA^*(x_k))\quad \text{ for } x_1,\ldots,x_k\in\partial a \subset V , \psi\in \Psi.
\end{align*}
The students get to see the random graph $\G^*$ but not the ground truth and stand before the task of deciphering as much information about $\SIGMA^*$ as possible. Thus the limit on the amount of deductable information is quantified by the mutual information $I(\SIGMA^*,\G^*)$. 

\begin{theorem}
\label{thm:mainInfThm}
With $\vec\gamma$ chosen from $D$, $\vec\psi_1,\vec\psi_2,\ldots\in\Psi$ chosen from $p$, $\vec\mu_1^{(\pi)},\vec\mu_2^{(\pi)},\ldots$ chosen from $\pi\in\cP_*^2(\Omega)$ and $\vec h_1,\vec h_2,\ldots\in[k]$ chosen uniformly at random, all mutually independent let
\begin{align*}
&\textstyle\cB(D,\pi)=\frac{1}{|\Omega|}\Erw\left[\xi^{-\vec \gamma}\Lambda\left(\sum_{\sigma\in\Omega}\prod_{i=1}^{\vec\gamma}\sum_{\tau\in\Omega^k}\vecone\{\tau_{\vec h_i}=\sigma\}\PSI_i(\tau)\prod_{j\neq \vec h_i}\mu^{(\pi)}_{ki+j}(\tau_j)\right)\right]\\
&\textstyle\qquad\qquad\quad-\frac{k-1}{k\xi}\Erw[\vec\gamma]\Erw\left[\Lambda\left(\sum_{\tau\in\Omega^k}\vec\psi(\tau)\prod_{j=1}^k\vec\mu_j^{(\pi)}(\tau_j)\right)\right].\nonumber
\end{align*}
Let $(\SIGMA^*,\G_D^*)$ be an assignment/factor graph pair from the teacher-student model. Then
\begin{align*}
\lim_{n\to\infty}\frac{1}{n}I(\SIGMA^*,\G_D^*) = - \sup_{\pi\in\PP}\cB(D,\pi) + \ln |\Omega| + \frac{\Erw[\vec{\gamma}]}{k\xi|\Omega|^k}\sum_{\tau\in\Omega^k}\Erw[\Lambda(\PSI(\tau))].
\end{align*}
\end{theorem}
The mutual information in the teacher-student model has been previously derived for \Erdos-\Renyi~ type LDGM graph ensembles in \cite{CKPZ}. We include standard $(D,k)$-graph ensembles $\G^*_D$ by a delicate extension of their techniques to a modified Poisson approximation inspired by Montanaris \emph{Multi-Poisson ensembles} \cite{montanari}. The technical proofs involved will be carried out for arbitrarily precise approximations $\G^*_{\alpha,\beta,D}$ of random $(D,k)$-graphs $\G^*_D$ and then bridged to the original model by way of a coupling. The following Proposition establishes the precision of our approximative model.
\begin{proposition}
\label{prop:goodcoupling}
Let $\rG_n = \rG^*_{\alpha,\beta,D}$ and $\rG'_n = \G^*_{D}$. There is a bounded function $f:[0,1)\times (0,\infty)\to\RR$ and a coupling of $\rG_n$ and $\rG'_n$, such that if we add the missing $O(\alpha n)$ check nodes in $\G_n$ the expected number of check nodes that have a different neighborhood in $\G_n'$ is $O(f(\alpha,\beta))$.
\end{proposition}
In other words, we can obtain a random $(D,k)$-graph $\G^*_D$ by first generating a Multi-Poisson approximation $\G^*_{\alpha,\beta,D}$ and then rewiring a small number of edges. We prove Proposition \ref{prop:goodcoupling} in Section \ref{sec:proofstrat} by coupling $\G^*_{\alpha,\beta,D}$ to the exact instance $\G^*_D$ chosen from the configuration model. We track the incremental influence of small perturbations during the construction process and prove that the approximation error is negligible. The coupling is similar to the procedure from \cite{montanari}, yet by introducing the additional approximation parameter $\beta$ we achieve an error that is asymptotically independent of $n$.

To outline the proof of Theorem \ref{thm:mainInfThm} suppose that $(\SIGMA^*,\G^*_D)$ is a jointly distributed tuple from the teacher-student model, where the underlying factor graph is a random $(D,k)$-graph. As the mutual information can be expressed by the marginal entropies, we can write 
\begin{align*}
I(\SIGMA^*,\G_D^*)=H(\SIGMA^*)-H(\SIGMA^*|\G_D^*).
\end{align*}
While $H(\SIGMA^*)$ is easily calculated the derivation of $\cH=H(\SIGMA^*|\G_D^*)$ requires a more sophisticated approach, making use of several methods from the physicist's toolbox for sparse mean-field spin-glasses. Apart from inherent properties of the teacher-student model we will prove and make use of the so called \emph{Nishimori property} from physics. This property allows us to identify the reweighed a posteriori measure induced by the graph, commonly named the \emph{Gibbs measure} $\mu_{\G_D^*}$, with samples from the the ground truth. Results from \cite{CKPZ} imply that given the Nishimori property we can perform slight perturbations on $\mu_{\G_D^*}$ that without significant impact on the mutual information allow us to factorize its marginals. This novel \emph{pinning technique} is the crucial new ingredient that facilitates the proof of Montanaris conjecture.

Therewith we have layed a foundation to derive an upper bound on the conditional entropy, i.e., a lower bound on the mutual information, by means of the so-called \emph{Aizenman-Sims-Starr} scheme. The proof is carried out in Section \ref{sec:lowerbound}. In essence the procedure boils down to estimating the change in entropy if we go from a model with $n$ variable nodes to one with $n+1$ variable nodes (as indicated in Figure \ref{fig:ass}). 

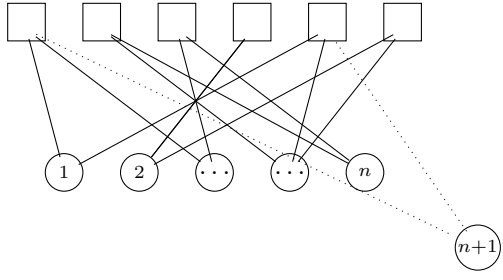
\begin{figure}[ht]
\begin{tikzpicture}

\draw (0.25,2.75) rectangle (0.75,3.25) (0.5,2.9) node (factor1) {};
\draw (1.25,2.75) rectangle (1.75,3.25) (1.5,2.9) node (factor2) {};
\draw (2.25,2.75) rectangle (2.75,3.25) (2.5,2.9) node (factor3) {};
\draw (3.25,2.75) rectangle (3.75,3.25) (3.5,2.9) node (factor4) {};
\draw (4.25,2.75) rectangle (4.75,3.25) (4.5,2.9) node (factor5) {};
\draw (5.25,2.75) rectangle (5.75,3.25) (5.5,2.9) node (factor6) {};

\draw (1,1) circle (0.25) node (var1) {$\scriptstyle 1$};
\draw (2,1) circle (0.25) node (var2) {$\scriptstyle 2$};
\draw (3,1) circle (0.25) node (var3) {$\ldots$};
\draw (4,1) circle (0.25) node (var4) {$\ldots$};
\draw (5,1) circle (0.25) node (var5) {$\scriptstyle n$};
\draw (6.5,0) circle (0.3) node (var6) {$\scriptstyle n+1$};

\draw (var2) -- (factor6);
\draw (var4) -- (factor6);

\draw (var1) -- (factor5);
\draw (var4) -- (factor5);

\draw (var2) -- (factor4);
\draw (var2) -- (factor4);

\draw (var3) -- (factor3);
\draw (var5) -- (factor3);

\draw (var4) -- (factor2);
\draw (var5) -- (factor2);

\draw (var1) -- (factor1);
\draw (var3) -- (factor1);

\draw [dotted] (var6) -- (factor5);
\draw [dotted] (var6) -- (factor1);
\end{tikzpicture}
    \caption{A change of entropy is induced by introducing another variable into the random factor graph model.}
    \label{fig:ass}
\end{figure}

Denote by $\Delta_{\cH}(n)$ the change in conditional entropy when going from $n$ to $n+1$.The representation $
\frac1n\cH_n = \frac1n\sum_{j=0}^{n-1}\left(\cH_{n+1}-\cH_n\right)=\frac1n\sum_{j=0}^{n-1}\Delta_{\cH}(n)
$ clearly implies that if $\Delta_{\cH}(n)$ converges, then its limit is also the limit of $\cH_n$. Unfortunately, we will not be able to compute the limit directly, but we can settle by using the representation to obtain an upper bound on the entropy in the limit
\begin{align*}
\limsup_{n\to\infty}\cH_n \le \limsup_{n\to\infty} \Delta_{\cH}(n).
\end{align*}

This is achieved by generating our graphs from a random graph process and rigorously coupling the underlying models. The result of an explicit calculation carried out in Section \ref{sec:lowerbound} is that $\Delta_{\cH}(n)$ can itself be upper bounded by the conjectured formula. This translates into the following lower bound on the mutual information.

\begin{proposition}
\label{prop:strategy_lowerbound}
We have
\begin{align*}
\liminf_{n\to\infty}\frac{1}{n}I(\SIGMA^*,\G_D^*) \ge - \sup_{\pi\in\PP}\cB(D,\pi) + \ln |\Omega| + \frac{\Erw[\vec{\gamma}]}{k\xi|\Omega|^k}\sum_{\tau\in\Omega^k}\Erw[\Lambda(\PSI(\tau))].
\end{align*}
\end{proposition}

In Section \ref{sec:upper_bound}, using the Guerra-Toninelli \emph{interpolation method} we derive a matching lower bound on the conditional entropy. We interpolate between the original model and a much simpler graph model, where the Gibbs measure is more easily understood and we can effortlessly verify that the entropy coincides with the conjectured formula. The random graph approximation is assembled in a sequence of layers, such that depending on the approximation parameter $\beta>0$ each layer contributes $\Po(\beta)$ many check nodes. We refine the interpolation argument by segmentation into layers, as portrayed by Figure \ref{fig:interpolation}.

\begin{figure}[ht]
\begin{tikzpicture}

\draw (0.25,2.75) rectangle (0.75,3.25) (0.5,2.9) node (factor1) {};
\draw (1.25,2.75) rectangle (1.75,3.25) (1.5,2.9) node (factor2) {};
\draw (2.25,2.75) rectangle (2.75,3.25) (2.5,2.9) node (factor3) {};
\draw (3.25,2.75) rectangle (3.75,3.25) (3.5,2.9) node (factor4) {};

\draw (4.5,3.5) -- (4.5,0.5);

\draw (5.25-0.5,2.75) rectangle (5.75-0.5,3.25) (5.5-0.5,2.9) node (factor5) {};
\draw (6-0.5,2.75) rectangle (6.5-0.5,3.25) (6.25-0.5,2.9) node (factor6) {};
\draw (6.75-0.5,2.75) rectangle (7.25-0.5,3.25) (7-0.5,2.9) node (factor7) {};

\draw (6.5+0.5,2.75) rectangle (7+0.5,3.25) (6.75+0.5,2.9) node (factor8) {};

\draw (7.75,3.5) -- (7.75,0.5);

\draw (7.5+0.5,2.75) rectangle (8+0.5,3.25) (7.75+0.5,2.9) node (factor9) {};
\draw (8.75,2.75) rectangle (9.25,3.25) (9,2.9) node (factor10) {};

\draw [dotted] (9.75,1.5) -- (10.25,1.5);

\draw (10.25,2.75) rectangle (10.75,3.25) (10.5,2.9) node (factor11) {};

\draw (0,1) circle (0.25) node (var1) {};
\draw (1,1) circle (0.25) node (var2) {};
\draw (2,1) circle (0.25) node (var3) {};
\draw (3,1) circle (0.25) node (var4) {};
\draw (4,1) circle (0.25) node (var5) {};
\draw (5.5-0.5,1) circle (0.25) node (var6) {};
\draw (6.25-0.5,1) circle (0.25) node (var7) {};
\draw (7-0.5,1) circle (0.25) node (var8) {};
\draw (7.25,1) circle (0.25) node (var9) {};
\draw (8.25,1) circle (0.25) node (var10) {};
\draw (9,1) circle (0.25) node (var11) {};
\draw (10.5,1) circle (0.25) node (var12) {};

\draw (var1) -- (factor4);
\draw (var2) -- (factor4);
\draw (var6) -- (factor4);

\draw (var3) -- (factor3);
\draw (var5) -- (factor3);
\draw (var4) -- (factor3);

\draw (var4) -- (factor2);
\draw (var5) -- (factor2);
\draw (var2) -- (factor2);

\draw (var1) -- (factor1);
\draw (var3) -- (factor1);
\draw (var5) -- (factor1);

\draw [dashed] (var5) -- (factor5);
\draw [dashed] (var7) -- (factor6);
\draw [dashed] (var8) -- (factor7);

\draw [dashed] (var4) -- (factor8);
\draw [dashed] (var6) -- (factor8);
\draw [dashed] (var9) -- (factor8);

\draw (var10) -- (factor9);
\draw (var11) -- (factor10);
\draw (var12) -- (factor11);
\node [label={[xshift=4.75cm, yshift=0cm]$s$}] {};
\node [label={[xshift=8.25cm, yshift=0cm]$\scriptstyle s+1$}] {};
\end{tikzpicture}
    \caption{The interpolation scheme for $k=3$ during $t\in(0,1)$ in layer $s$, originally consisting of $\Po(\beta)=2$ factor nodes (dashed neighborhoods). One of two $k$-ary factor nodes in layer $s$ is split into $k$ unary nodes.}
    \label{fig:interpolation}
\end{figure}
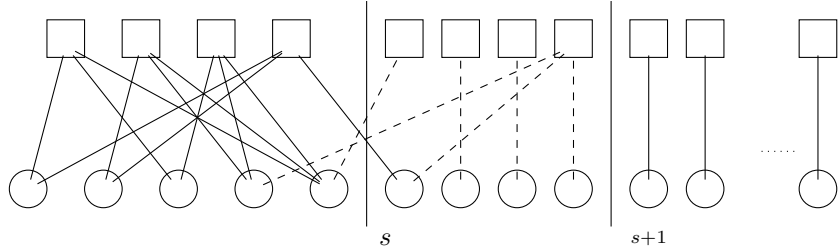

Beginning with the surface layer of $\G_1=\G_D$ we split its $k$-ary factor nodes one at a time, replacing them by unary nodes with weight functions simulating the complex underlying structure. We thus peel apart the intricate composition of $\G_D$ into a forest $\G_0$ of unary variable components. By independently splitting each $k$-ary factor node of layer $s$ with probability $t\in[0,1]$, sequentially for each layer $s$, we ensure a continuous interpolation between $\G_1$ and $\G_0$. Ultimately, we aim to show that the entropy in the latter model does indeed upper bound the original conditional entropy. This amounts to controlling $\frac{\partial}{\partial t} \cH_n(\G_t)$ within the interpolation interval and involves yet another clever coupling argument, carried out in full detail in section \ref{sec:upper_bound}. Taken as a whole, because the derivative with respect to $t$ is positive during the entire interpolation, we obtain the desired lower bound on the conditional entropy and as such the following upper bound on the mutual information.
\begin{proposition}
\label{prop:strategy_upperbound}
\begin{align*}
\limsup_{n\to\infty}\frac{1}{n}I(\SIGMA^*,\G_D^*) \le - \sup_{\pi\in\PP}\cB(D,\pi) + \ln |\Omega| + \frac{\Erw[\vec{\gamma}]}{k\xi|\Omega|^k}\sum_{\tau\in\Omega^k}\Erw[\Lambda(\PSI(\tau))].
\end{align*}
\end{proposition}

Together, Propositions \ref{prop:strategy_lowerbound} and \ref{prop:strategy_upperbound} immediately imply the assertion for any planted CSP ensemble with given variable degree distributions. Theorem \ref{thm:Codes} follows by explicitly calculating the quantities in the case of LDGM codes, and verifying that \textbf{SYM} and \textbf{POS} hold. These simple but technical computations are put off to Section \ref{sec:proofstrat}.

\section{The Teacher-Student model and Symmetry}
\label{sec:proofstrat}
In this section we set up the tools utilized in our proof. We formalize the notion of retrieval by constructively defining the teacher-student model for $(D,k)$-graph ensembles as well as their Poissonian approximation. To this end we verify a number of expedient properties brought with the model in our setting. We analyze the approximative random graph model in terms of its \emph{free energy}. This quantity is closely related to the mutual information and the main object of interest in many physics models of spin-glasses. Moreover, to derive tight bounds from the expressions we obtain for the free energy, we employ the \emph{Pinning Lemma} from \cite{CKPZ}. Finally we show how bounds on the free energy imply our general main result.  
Following this section, the lower and upper bound are derived in Sections \ref{sec:lowerbound} and \ref{sec:upper_bound} respectively.

\subsection{Preliminaries and Notation}
With a factor graph $G=(V,F,(\partial a)_{a\in F},(\psi_a)_{a\in F})$ we associate its \emph{partition function} $
Z(G)=\sum_{\sigma\in\Omega^n}\psi_G(\sigma)$, where $\psi_G(\sigma)=\prod_{a\in F}\psi_a(\sigma(\partial_1 a,\ldots,\partial_k a))
$. This gives rise to the \emph{Gibbs distribution} $\mu_G(\sigma) = Z(G)^{-1}\psi_G(\sigma)$. If $X:\Omega^n\to\RR$ is a functional, we write $\langle X(\SIGMA) \rangle_{G}=\sum_{\sigma\in\Omega^n}\mu_G(\sigma)X(\sigma)$, when referring to its average with regard to the Gibbs distribution. Additionally, we define the empirical distribution on the graph as the average Gibbs marginal $
	\pi_G=|V|^{-1}\sum_{x\in V}\delta_{\mu_{G,x}}.
$
If $\mu$ is a probability measure, we write $\mu^{\otimes \ell}$ to denote its $\ell$-fold product measure. Moreover, we write $\vec\mu^{(\pi)}$ for a random sample from $\pi\in\cP^2_*(\Omega)$.

The standard graph ensembles follow a predefined left-degree distribution. Considering degree sequences chosen from those distributions will be an expedient tool for the construction of random graphs from a given ensemble. We call a degree sequence $d$ finite, if $\sup_x d(x)<\infty$. If not explicitly stated otherwise, we assume the number of variable nodes in a given graph to be $n$ and degree sequences to be finite with maximal degree $\Delta$. For a real number $x$ we let $(x)_+=\max\{0,x\}$. 

Before we formally specify the random $(D,k)$-graph construction process, let us define a process that samples graphs with specific left-degree sequences. This is commonly known as the \emph{configuration model}.
\begin{definition}
\label{exactProcedure} 
Let $d:[n]\to\N$ be a degree sequence. By $\bar\G_{n,d}$ we denote the random  unweighted factor graph obtained from the following process.
\begin{description}
\item[EX1] Let $F=\emptyset$, $s=1$. Initiate a vector $\vec \delta_s$ by setting $\vec \delta_s(x) = d(x)$ for $x=1,\ldots, n$. 
\item[EX2] While $\vec\delta_s \neq 0$:
Choose a random $\vec x_s\in[n]$ from the measure $\vec\nu_s$ defined by
$$\vec\nu_s(x)=\frac{\vec \delta_s(x)}{\sum_{y\in[n]}\vec \delta_s(y)}.$$
Update $\vec \delta_{s+1}=\vec \delta_s-(\vecone\{\vec x_s\})_{x\in [n]}$ and set $s=s+1$.
If $s-1$ is a multiple of $k$, then add a new check node $\vec a$ to $F$ and set its neighborhood to be $\{\vec x_{s-k-1},\ldots,\vec x_{s-1}\}$.
\end{description}
The resulting graph is $\bar\G_{n,d}=([n],F,(\partial \vec a)_{\vec a\in F})$ and satisfies the degree sequence as long as $n$ divides $\sum_{x\in[n]}d(x)$. We also write $\bar\G_{d}$ if the number of variable nodes $n$ is apparent from the context.

\end{definition}

\begin{definition}
If $[n]$ allows for a partition $(V_l)_{l}$ with $|V_l|= nD(l)$, a random unweighted factor graph from the $(D,k)$-ensemble can be obtained by choosing such a partition $(\vec V_l)_{l}$ uniformly at random, setting $\vec d(x)=l$ for $x\in \vec V_l$, $l=1,2,\ldots$. We then write $\G_D$ for a random choice from $\bar\G_{\vec d}$.
\end{definition}

At times it is beneficial to generate the random graph in a sequence of batches of factor nodes according to a specified structure, rather than one node-socket at a time. This is facilitated by means of the following model that we will then use to lay down the Poissonian approximation. 

\begin{definition}
	\label{def:approx}
Given a degree sequence $d$ as well as a vector $m=(m_1,\ldots,m_{\sm})$ with non-negative integer entries, we define $\G_{n,m,d}$ to be the random graph obtained from the following experiment. For the sake of clarity we omit the subscript $n$ if it is unambiguous.
\begin{description}
\item[AP1] Let $F=\emptyset$, $s=1$. Initiate a vector $\vec \delta_s$ by setting $\vec \delta_s(x) = d(x)$ for $x=1,\ldots, n$. 
\item[AP2] In layers $s=1,\ldots,\sm$ choose $m_s$ random neighborhoods $\partial \vec a_{s,i}$ from $\vec\nu_s^{\otimes k}$ 
and add factor nodes $\vec a_{s,1},\ldots, \vec a_{s,m_s}$ with respective neighborhoods to the graph. Denote by $\nabla_s(x)$ the number of times the variable node $x$ has been chosen as a neighbor in round $s$ and update $\vec \delta_{s+1}=(\vec\delta_s - \nabla_s)_+$, $s=s+1$.
\end{description}
\end{definition}

Given a degree distribution $D$ we say that a degree sequence $\vec d$ is a random $D$-partition of $[n]$ if for a uniformly random variable $\vec x\in[n]$ and any $\ell\ge 0$ the degrees satisfy $\pr[\vec d(\vec x)=\ell]=D(\ell)$.

\begin{definition}
Given a degree distribution $D$ and a vector $m=(m_1,\ldots,m_{\sm})$ with non-negative integer entries, we let $\vec d$ be a random $D$-partition of $[n]$ and thereby define $\G_{m,D}=\G_{m,\vec d}$.
\end{definition}

\subsection{The approximation scheme}
To obtain an approximation of $\G_D$ we introduce two additional parameters $\alpha\in[0,1)$, $\beta>0$ that aid in laying down an appropriate approximation-step vector $m$. 
\begin{definition}
\label{def:abApprox}
When $\alpha,\beta$ and a degree sequence $d$ are fixed let $\sm =\lfloor (1-\alpha)\beta^{-1}k^{-1}\sum_{x\in [n]}d(x)\rfloor$. Further, let $\vec m_1, \vec m_2,\ldots$ be a sequence of independent $\Po(\beta)$ distributed numbers. We then obtain an $(\alpha,\beta)$-approximation of $\bar\G_d$ by setting $\vec m = (\vec m_1,\ldots,\vec m_{\sm})$ and letting $\G_{\alpha,\beta,d}=\G_{\vec m,d}$. Analogously we define $\G_{\alpha,\beta,D}=\G_{\vec m, D}$, where we first choose a random $D$-partition $\vec d$ and generate $\vec m$ thereafter.
\end{definition}

Note that the graph-generating procedure as stated in Definition \ref{def:approx} will fail if $\sum_{x \in V} \rDelta_s(x) = 0$ for some $1 \le s \le s_{\max}$. However, it will succeed \whp~ as $n$ tends to infinity because the probability of adding more than $k^{-1}(1-\alpha)\sum_{x \in V} d(x)$ factor nodes during the construction tends towards zero. Particularly, as $\alpha,\beta$ tend to 0 this becomes an arbitrarily close approximation of $\bar\G_d$, in the sense that we obtain $\bar\G_{d}$ from $\G_{\alpha,\beta,d}$ by rewiring $O(1)$ node sockets. This is enabled by the following observation.

\begin{fact}\label{deltaTVlemma}
With $\alpha,\beta$ fixed and $s\in[\sm]$ let $(\vec \delta_s(v))_{v\in[n]}$ be the degree sequence from Definitions \ref{exactProcedure} and \ref{def:approx}. Let $c:[n]\to\ZZ$ be a bounded integer valued function on the variable nodes and let $\rDelta'_s(v) = (\rDelta_s(v)-c(v))_+$. Consider distributions $\vec \nu_s\propto\vec\delta_s$ and $\vec\nu_s'\propto\vec\delta'_s$. Then
\begin{align*}
\|\vec\nu_s-\vec\nu_s'\|_{TV}
= \frac{\sum_{x\in V} [c(x)]_+}{O(\alpha n)-\sum_{x\in V} c(x)}.
\end{align*}
\end{fact}
\begin{proof}
We simply write out the total variation distance and bound $\sum_x\vec\delta_s(x)\ge\sum_x\vec\delta_{\sm}(x)=\Omega(\alpha n)$ to obtain
\begin{align*}
&\|\vec\nu_s-\vec\nu_s'\|_{TV}=\sum_{v\in [n]}\left(\vec\nu_s(v)-\vec\nu'_s(v)\right)_+ = \sum_{v\in [n]} \left(\frac{\vec\delta_s(v)}{\sum_{x\in [n]}\vec\delta_s(x)}-\frac{\vec\delta_s(v)-c(v)}{\sum_{x\in [n]}\vec\delta_s(x)-\sum_{x\in [n]}\vec\delta_c(x)}\right)_+\\
=&\sum_{v\in [n]}\left(\frac{c(v)\sum_{x\in[n]}\vec\delta_s(x)-\vec\delta_s(v)\sum_{x\in[n]}c(x)}{\sum_{x\in[n]}\vec\delta_s(x)-\sum_{x\in[n]}c(x)}\right)_+\le \sum_{v\in[n]}\left(\frac{c(v)}{\sum_{x\in [n]}(\vec\delta_s(x)-c(x))}\right)_+.
\end{align*}
\end{proof}

In the light of Lemma \ref{deltaTVlemma} consider generating an approximative graph $\G_{\alpha,\beta,D}$. If we independently generate the missing $O(\alpha n)$ factor nodes in the same fashion as the exact procedure from Definition \ref{exactProcedure}, we can verify that the distribution of the resulting factor graph matches the distribution of ${\G}_D$ after we rewire a finite number of factor nodes. Importantly, this number is asymptotically independent of $n$.

\begin{proof}[Proof of Proposition \ref{prop:goodcoupling}]
If $\beta$ is sufficiently small, most layers will contain at most one check node, and very few will contain a constant number of check nodes. For the sake of simplicity suppose that each layer consists of at most one check node. The argument extends to any constant number of check nodes. 

First observe that during the layer-wise creation of $\G_n$ in any step of layer $s$, if the approximative graph matches the original graph except for $C>0$ factor nodes, there are at most $Ck$ different entries in the degree sequences. That is  $\|\vec\nu_s-\vec\nu_s'\|_{TV}=O(\frac{Ck}{\alpha n})$ by Fact \ref{deltaTVlemma}. As the degree sequence in the exact model updates immediately, the approximative degree sequences adds up to $k-1$ incremental deviations during the creation of one $k$-ary check node. This results in a total variation distance of $kO(\frac{(C+1)k}{\alpha n-k})=O(\frac{C}{\alpha n})$ for the choice of the next neighbor. We will now couple the configuration model process in both graphs inductively. Suppose that during the process, we have an optimal coupling such that both graphs differ in exactly $i-1\ge0$ check nodes. By the coupling Lemma both graphs coincide with probability $O(\frac{i}{\alpha n})$, thus we can choose the same neighborhood for the next $\vec X_i$ check nodes, where $\vec X_i$ is a geometrically distributed random variable with parameter $p_i=O(\frac{i}{\alpha n})$. Hence, once the process is complete the total number $\vec C_F$ of different check nodes is the minimum number of independent random variables $\vec X_1,\vec X_2,\ldots$ such that $\vec X_i$ is chosen from the geometric distribution with $p_i=O(\frac{i}{\alpha n})$ and $\sum_{i=1}^{\vec C_F}\vec X_i\ge \Omega((1-\alpha)n)$. As $n$ tends to infinity, the random variables $\frac{\vec X_i}{\alpha n}$ can be arbitrarily well approximated by random variables $\vec Y_i$ chosen from the exponential distribution with parameter $\hat p_i = O(i)$, such that $\sum_{i=1}^{\vec C_F} \vec Y_i \ge \Omega(\frac{1-\alpha}{\alpha})$. As a consequence $\vec C_F$ can be modeled as a function of $\alpha$ and is asymptotically independent of $n$.

We have assumed that each layer $s=1,\ldots, \sm$ consists of at most one check node. The argument generalizes to any sequence of layers with independent $\Po(\beta)$ many check nodes. Because the tails of the Poisson distribution show sub-exponential decay, we can guarantee that $\|\nu_s-\nu_s'\|=O(\frac{C+\beta}{\alpha n})$, if in layer $s-1$ the coupling differs in at most $C$ check nodes. By induction we can write $\vec C_F$ as a function of $\alpha$ and $\beta$.
\end{proof}

\subsection{The Teacher-Student Model}
To derive the mutual information between the transmitted message and the received one, we analyze the free energy density within the $(D,k)$-teacher-student model. To this end, we generalize the scheme from \cite{CKPZ} to arbitrary degree distributions. We define a random factor graph $\G^*$ via the following construction.

\begin{description}
\item[TCH1] Choose $\SIGMA^*:[n]\to\Omega$ uniformly at random.
\item[TCH2] Choose a random graph $\G=([n],F,(\partial a)_{a\in F})$ from $\G_{\alpha,\beta,D}$.
\item[TCH3] For each check node $a\in F$ choose a random $\vec\psi_a$ in $\Psi$ from the distribution
\begin{align*}
	\pr[\vec\psi_a =\psi]=\frac{p(\psi)\psi(\SIGMA^*(\partial_1 a),\ldots,\SIGMA^*(\partial_k a))}{\sum_{\psi'\in\Psi}\psi'(\SIGMA^*(\partial_1 a),\ldots,\SIGMA^*(\partial_k a))}
\end{align*}
We let $\G^*=\G^\ast(\SIGMA^*)=([n],F,(\partial a)_{a\in F},(\vec\psi_a )_{a\in F})$ denote the resulting (weighted) random factor graph and omit the parameters $\alpha,\beta$ and $n$ for the sake of readability. We write $\G^*_D$ for the limit $\alpha,\beta\to 0$.
\end{description}

\subsection{Pinning and Symmetry}
Another indispensable tool that we make use of is \emph{pinning}. As we derive bounds on the free energy $\Erw[\ln Z(\G^*)]$ we are confronted with drawing $k$-tuples of bit-assignments chosen from the Gibbs measure. Due to the confinement by weight functions certain configurations of  tuples are favored within the Gibbs measure. Hence, in most cases variable node assignments sampled from the Gibbs measure, where the nodes are connected through a common check node will be highly correlated by means of the weight function and thus far from independent. This correlation ought to persist for any nodes within a finite distance in the graph. On the other hand if we choose two variable nodes at random they are typically far apart, which is why we can hope that their Gibbs measure is almost a product measure. This is captured by the notion of $\eps$-symmetry, or more generally $(\delta,\ell)$-symmetry.

\begin{definition}
	Let $\mu_G$ be a probability measure on $\Omega^n$. For $\ell\ge 2$ and $\{x_1,\ldots,x_\ell\}\subset [n]$ let $\mu_{G,x_1,...x_\ell}$ denote the marginal distribution of an $\ell$-tuple chosen from $\mu_G$. 
	We say that $\mu_G$ is $(\delta,\ell)$-symmetric, if 
	\begin{align*}
	\frac{1}{n^\ell} \sum_{x_1,...,x_l\ell\in [n]} \| \mu_{G,x_1,...,x_\ell} - \mu_{G,x_1} \otimes \cdots \otimes \mu_{G,x_\ell} \|_{TV} < \delta.
	\end{align*}
	If $\mu_G$ is $(\eps,2)$-symmetric, we simply speak of $\eps$-symmetry.
\end{definition}

Fortunately, the following \emph{Pinning Lemma} from \cite{CKPZ} guarantees a degree of $\eps$-symmetry in exchange for small modifications in the original measure. 

\begin{lemma}{\cite{CKPZ}}
\label{lem:PinningLemma}
For any $\eps>0$ there is $T=T(\eps,\Omega)>0$ such that for every $n>T$ and every probability measure $\mu\in\cP(\Omega^n)$ the following is true. Obtain a random probability measure 
$\vec{\check\mu} \in\cP(\Omega^n)$ as follows: 
Draw a sample $\check\SIGMA\in\Omega^n$ from $\mu$, independently choose a number $\vec\theta\in(0,T)$ uniformly at random and obtain a random set $\vec U\subset [n]$ by including each $i\in [n]$ with probability $\vec\theta/n$ independently. The measure $\vec{\check\mu}$ defined by
\begin{align*}
\vec{\check\mu}(\sigma)=\frac{\mu(\sigma)\vecone\{\forall i\in \vec U:\sigma_i =\check\SIGMA_i\}}{\mu(\{\tau\in\Omega^n :\forall i\in \vec U:\tau_i =\check\SIGMA_i\})},\qquad\sigma\in\Omega^n
\end{align*}
is $\eps$-symmetric with probability at least $1-\eps$.
\end{lemma}

When applying the procedure from Lemma \ref{lem:PinningLemma} to the Gibbs measure of a graph the perturbation equates to pinning the color of each variable node $i$ in $\vec U$ to its coloring under $\check\SIGMA$, i.e. adding a constraint $\vec\psi_i(\sigma)=\vecone\{\sigma = \check\SIGMA_i\}$ to those variable nodes.

\begin{definition}
\label{def:pinning}
Given a factor graph $G$, a subset $U\subset [n]$ of its variable nodes and an assignment $\check\sigma\in\Omega^n$ obtain $G_{U,\check\sigma}$ from $G$ by adding unary check nodes $a_x$ with weights $\psi_{a_x}(\tau)=\vecone\{\check\sigma(x)=\tau\}$ to each variable node $x\in U$. Moreover, for $T\ge 0$ let $\vec U = \vec U(T)\subset V$ be a random subset of vertices generated by first choosing $\vec \theta\in[0,T]$ uniformly at random and then including each variable node $v\in[n]$ in $\vec U$ with probability $\vec \theta/n$. If $\SIGMA\in\Omega^n$ is chosen independently and uniformly at random from $\Omega$ for each variable node, we let $G_T=G_{\vec U,\SIGMA}$.
\end{definition}

The following Lemma establishes, that the perturbations performed in the pinning procedure yield an $\eps$-symmetric measure with probability at least $1-\eps$.

\begin{lemma}[{\cite[Lemmata 2.8 and 3.5]{CKPZ}}]
\label{lem:epsSymmetry}
Let $\mu_G$ be a Gibbs measure on $\Omega^n$.
\begin{enumerate}[1.]
\item
For any $\eps>0$ there is a $T_0=T_0(\Omega,\eps)$ such that for all $T>T_0$ and sufficiently large $n$ the meausre $\mu_{G_T}$ is $\eps$-symmetric with probability at least $1-\eps$.
\item Moreover, for any $l\ge 3$, $\delta>0$ there is $\eps>0$ such that if $\mu_{G_T}$ is $\eps$-symmetric, then $\mu$ is $(\delta,l)$-symmetric.
\end{enumerate}
\end{lemma}

Also note that the pinning procedure has merely constant additive impact on the free energy. In the factor graph model all weight functions on check nodes are positive and all variable degrees are at most $\Delta$. Such being the case, pinning a single variable node changes the free energy of the graph by at most $O(1)$. Because the expected number of pinned variable nodes is at most $T$, we have the following Lemma.
\begin{lemma}
	If $\G$ is a random factor graph from the $(D,k)$-ensemble, then
	\begin{align*}
	\Erw[\ln Z(\G_T)]=\Erw[\ln Z(\G)]+O(1).
	\end{align*}
\end{lemma}

Lemma \ref{lem:epsSymmetry} together with the so-called Nishimori property will be key tools utilized in the proof. The Nishimori property stands testament to the fact that if the prior $p$ satisfies a certain symmetry condition, the Gibbs measure in the teacher student model $\mu_{\G^*}$ does indeed resemble the a posteriori distribution of our planted assignment $\SIGMA^*$ given the graph outcome $\G^*$. We will use this argument in the same way as done in the proofs of \cite{CKPZ}. Nevertheless, the symmetry condition from \cite{CKPZ} imposes a much weaker restriction on the prior $p$. While being a direct implication of \textbf{SYM}, their condition merely yields mutual contiguity of both measures. In the class $(\Psi,p)$ of models that covers error-correcting codes transmitted over binary memoryless channels, we obtain the following stronger formulation.

\begin{lemma}
\label{lem:nishimori}
Suppose that $p$ satisfies \textbf{SYM}. Then
\begin{align}\label{eq:Nishi1}
\pr[\SIGMA^*=\sigma, \G^*(\SIGMA^*)=G]=\pr[\G^*= G]\mu_{G}(\sigma).
\end{align}
In particular, the distribution of $\SIGMA^*$ coincides with the Gibbs measure $\mu_{\G^*}$.
\begin{proof}
If $\G$ is a random graph model and the event $|F(\G)|=M$ has positive probability, we write $\G_M$ to be the conditional random graph on $M$ factor nodes. Observe that for any choice of $M, \sigma$ and any event $\cA$
\begin{align}
\pr[\G^*_M(\SIGMA^*)\in\cA|\SIGMA^*=\sigma]=\frac{\Erw[\psi_{\G_M}(\sigma)\vecone\{\G_M\in\cA\}]}{\Erw[\psi_{\G_M}(\sigma)]}.\label{eq:reweigh}
\end{align}
To see \eqref{eq:reweigh} it suffices to write out the weights for any event $$\cA=\{(V,F,(\partial a)_{a\in F},(\psi_a)_{a\in F}:(\psi_a)_{a\in F}\in W\subset \Psi^F\}.$$
Doing so, we get
\begin{align*}
&\pr[\G^*_M(\SIGMA^*)\in\cA|\SIGMA^*=\sigma]=\pr[\G_M=(V,F,(\partial a)_{a\in F})]\pr[(\vec \psi_a)_{a\in F}\in W|\SIGMA^*=\sigma]\\
=&\pr[\G_M=(V,F,(\partial a)_{a\in F})]\xi^{-M}\Erw[\psi_{(V,F,(\partial a)_{a\in F})}(\sigma)\vecone\{(\vec \psi_a)_{a\in F}\in W\}]\\
=&\xi^{-M}\Erw[\psi_{\G_M}(\sigma)\vecone\{\G_M\in\cA\}]
\end{align*}
and are left to verify that summing over all graphs $G=(V,F,(\partial a)_{a\in F})$ on $M$ factor nodes gives
\begin{align*}
&\Erw[\psi_{\G_M}(\sigma)]=\sum_{G}\pr[\G_M = G]\Erw[\psi_{\G(\sigma)}|\G = G] \\=& \sum_{G} \pr[\G_M=(V,F,(\partial a)_{a\in F})]\prod_{a\in F}\Erw[\PSI(\sigma(\partial a)]=\xi^M.
\end{align*}
Having established \eqref{eq:reweigh} let us now derive \eqref{eq:Nishi1}. If we write $\vec M$ for the distribution of $F(\G^*)$, then $\G^*_{\vec M}$ and $\G^*$ have the same distribution. Moreover, for any $M$ we have
\begin{align*}
\pr[\G^*_M\in\cA]=&\sum_{\sigma\in\Omega^n}\pr[\SIGMA^*=\sigma]\pr[\G^*_M(\SIGMA^*)\in\cA\in \cA|\SIGMA^*=\sigma]=\sum_{\sigma\in\Omega^n}\frac{\xi}{\xi}\frac{1}{|\Omega|^n}\frac{\Erw[\psi_{\G_M}(\sigma)\vecone\{\G_M\in\cA\}]}{\Erw[\vec\psi_{\G_M}(\sigma)]}\\
=&\sum_{\sigma\in\Omega^n}\frac{\Erw[\psi_{\G_M}(\sigma)]}{\Erw[Z(\G_M)]}\frac{\Erw[\psi_{\G_M}(\sigma)\vecone\{\G_M\in \cA\}]}{\Erw[\psi_{\G_M}(\sigma)]}=\frac{\Erw[Z(\G_M)\vecone\{\G_M\in \cA\}]}{\Erw[Z(\G_M)]}.
\end{align*}
Consequently, if we write $\SIGMA$ for a sample from the Gibbs measure $\mu_{\G^*}$, the joint distribution $(\SIGMA,\G^*)$ satisfies
\begin{align*}
\pr[\SIGMA=\sigma,\G_M^*(\SIGMA)=\cA]&=\frac{1}{\Erw[Z(\G_M)]}\Erw\left[\frac{\psi_{\G_M}(\sigma)}{Z(\G_M)}Z(\G_M)\vecone\{\G_M\in \cA\}\right]
=\pr[\SIGMA^*=\sigma]\pr[\G^*_M|\SIGMA^*=\sigma]
\end{align*}
by multiplying the denominator and numerator with $\Erw[\psi_{\G_M}(\sigma)]$ and applying \textbf{SYM} in the final step.
\end{proof}
\end{lemma}

The following lemma shows that applying the pinning process in any order preserves the Nishimori property.
\begin{lemma}
\label{lem:nishimoripinned}
For any set of vertices $U\subset [n]$ the following distributions on pairs of factor graphs and assignments are identical. 
\begin{enumerate}[(1)]
\item Choose $\SIGMA^{(1)}=\SIGMA^*$, then choose $\G^{(1)}=\G_{U,\SIGMA^{(1)}}^*\left(\SIGMA^{(1)}\right)$ and output $\left(\SIGMA^{(1)},\G^{(1)}\right)$.
\item Choose $\G=\G^*$. Choose $\SIGMA^{(2)}$ from $\mu_{\G}$, then $\G^{(2)}=\G_{U,\SIGMA^{(2)}}$ and output $\left(\SIGMA^{(2)},\G^{(2)}\right)$.
\item Choose $\G^{(3)}=\G^*(\SIGMA)_{U,\SIGMA}$. Then choose $\SIGMA^{(3)}$ from $\mu_{\G^{(3)}}$ and output $\left(\SIGMA^{(3)},\G^{(3)}\right)$.
\item Choose $\G'=\G^*$ and $\SIGMA'$ from $\mu_{\G'}$. Pin $\G^{(4)}=\G_{U,\SIGMA'}$ and choose $\SIGMA^{(4)}$ from $\mu_{\G^{(4)}}$. Output $\left(\SIGMA^{(4)},\G^{(4)}\right)$.
\end{enumerate}
\end{lemma}
\begin{proof}
By Lemma \ref{lem:nishimori} the pairs (1) and (2) as well as (3) and (4) are identical. To see that (2) and (4) coincide, it suffices to prove that
\begin{align*}
\pr[\SIGMA^{(2)}=\sigma|\G^{(2)}\in\cA]=\pr[\SIGMA^{(4)}=\sigma|\G^{(4)}\in\cA]
\end{align*}
holds for any event $\cA=\{(G,\check\sigma)\,:\,G\in \cG\}$, where $\cG$ is an arbitrary set of unpinned factor graphs and $\check\sigma\in\Omega^U$ is a pinning of the variables in $U$. Denote by $\SIGMA_G$ a sample from the Gibbs measure $\mu_G$ and write $\sigma_U$ for the restriction of $\sigma$ to $U$. Then
\begin{align}
&\pr[\SIGMA^{(2)}=\sigma|\G^{(2)}\in\cA]=\pr[\SIGMA^{(2)}=\sigma|\G\in\cG, \SIGMA_U^{(2)}=\check\sigma]=\pr[\SIGMA_{\G}=\sigma|\G\in\cG,\SIGMA_{\G,U}=\check\sigma]\nonumber\\
=&\frac{1}{\Erw[\vecone\{\SIGMA_{\G,U}=\check\sigma\}\vecone\{\G \in\cG\}]}\Erw\left[\frac{\psi_{\G}(\sigma)\vecone\{\sigma_U=\check\sigma\}}{\sum_{\tau\in\Omega^n}\psi_{\G}(\tau)\vecone\{\tau_U=\check\sigma\}}\vecone\{\SIGMA_{\G,U}=\check\sigma\}\vecone\{\G\in\cG\}\right]\nonumber\\
=&\Erw[\vecone\{\G^{(2)}\in\cG\}]^{-1}\Erw\left[\frac{\psi_{\G^{(2)}}(\sigma)}{Z(\G^{(2)})}\vecone\{\G^{(2)}\in\cG\}\right].\label{eq:NM2and4}
\end{align}
By pinning the same neighborhood the evaluation in the teacher-student model remains the same and thus \eqref{eq:NM2and4} equals
$\Erw[\vecone\{\G^{(4)}\in\cG\}]^{-1}\Erw\left[Z(\G^{(4)})^{-1}\psi_{\G^{(4)}}(\sigma)\vecone\{\G^{(4)}\in\cG\}\right]=\pr[\SIGMA^{(4)}=\sigma|\G^{(4)}\in\cA].
$
\end{proof}

\subsection{Proof of the main result}
We will now prove the conjectured formula by showing that the mutual information per bit $\frac{1}{n}I(\SIGMA^*,\G_D^*)$ converges to the solution of a stochastic fixed-point equation. To state our result in a more general setting we will write the convergence in terms of the free energy $\frac{1}{n}\Erw[\ln Z(\G_D^*)]$

\begin{proposition}
\label{prop:upper_bound}
If \textbf{SYM} and \textbf{POS} hold, then $
\limsup_{n \to \infty} -\frac{1}{n} \E[\ln Z(\G_D^*)] \le - \sup_{\pi \in \P_*^2(\Omega)} \mathcal{B}(D,\pi).
$
\end{proposition}
We prove Proposition \ref{prop:upper_bound} in Section \ref{sec:upper_bound} by performing the Guerra-Toninelli interpolation between $\G_1=\G^*_D$ and a forest of isolated variable nodes $\G_0$. Taking account of the details that come along, this can be understood as simply splitting each $k$-ary factor node into $k$ unary factors with probability $1-t$, $t\in[0,1]$. As $-\Erw[\ln Z(\G_0)]$ is easy to compute, we can show that the resulting expression is also an upper bound for $-\Erw[\ln Z(\G_1)]$ by verifying that $\partial /\partial t \Erw[\ln Z(\G_1)]$ is positive on $(0,t)$. 

\begin{proposition}
\label{prop:lower_bound}
If \textbf{SYM} holds, then $
\limsup_{n \to \infty} -\frac{1}{n} \E[\ln Z(\G_D^*)] \ge - \sup_{\pi \in \P_*^2(\Omega)} \mathcal{B}(D,\pi).
$
\end{proposition}

In Section \ref{sec:lowerbound} we prove the lower bound in Proposition \ref{prop:lower_bound} via the previously described Aizenman-Sims-Starr scheme. That is we bound the difference between the free energy of $\G_n$ and $\G_{n+1}$. By coupling the joint distribution on a large common subgraph $\G^0$, we can compute the expected change in the free energy given by the additional constraints going from $\G^0$ to the graphs $\G_n$ and $\G_{n+1}$. This is possible due to Lemma \ref{lem:PinningLemma} as these constraints are generated from the empirical Gibbs marginals on $\G^0$ which does not significantly differ from $\pi_{\G_n}$ or $\pi_{\G_{n+1}}$.

By writing out the mutual information, Propositions \ref{prop:upper_bound} and \ref{prop:lower_bound} translate into Propositions \ref{prop:strategy_lowerbound} and \ref{prop:strategy_upperbound}, thus revealing the main theorem. 
\begin{proof}[Proof of Theorem \ref{thm:mainInfThm}]
	Theorem \ref{thm:mainInfThm} is a corollary of Propositions \ref{prop:upper_bound} and \ref{prop:lower_bound}. Writing out the mutual information with Lemma \ref{lem:nishimori} we get
	\begin{align*}
	I(\SIGMA^*,\G_D^*) = \sum_{G} \pr[\G_D^* = G]\sum_{\sigma}\mu_G(\sigma)\ln\frac{\mu_G(\sigma)}{\pr[\SIGMA^*=\sigma]} = H(\SIGMA^*)- \Erw[H(\mu_{\G^*})].
	\end{align*}
	As $\SIGMA$ chosen from $\mu_{\G^*}$ coincides with $\SIGMA^*$ given $\G^*$, we obtain
	\begin{align*}
	\Erw[H(\mu_{\G^*})]&=
	 -\Erw\left[\frac{\psi_{\G^*}(\sigma)}{Z(\G^*)}\ln \frac{\psi_{\G^*}(\sigma)}{Z(\G^*)}\right] = \Erw\ln Z(\G^*)- \Erw[\langle \psi_{\G^*}(\SIGMA)\rangle_{\G^*}] \\
	&=\Erw\ln Z(\G^*) - \Erw \ln \psi_{\G^*}(\SIGMA^*) = \Erw\ln Z(\G^*) - \frac{n\Erw[\vec \gamma]}{k|\Omega|^k}\sum_{\tau\in \Omega^k}\Erw\left[\Lambda(\PSI(\tau))\right].
	\end{align*}
	Because $\SIGMA^*$ is chosen from the uniform distribution on $\Omega^n$ the assertion follows.
\end{proof}

Theorem \ref{thm:mainInfThm} readily implies Theorem \ref{thm:Codes} by verifying \textbf{SYM} and \textbf{POS} for the case of $\eta$-noisy $(D,k)$-LDGM code.

\begin{proof}[Proof of Theorem \ref{thm:Codes}]
	Let $\Omega=\{\pm 1\}$, $\eta\in(0,1/2)$ and $(\Psi,p)=(\{\psi_1,\psi_{-1}\},p)$, where $p$ is the uniform distribution on $\Psi$ and $\psi_s(\sigma)=1+s(1-2\eta)\prod_{i=1}\sigma_i$ for $s\in\Omega$,  $\omega\in\Omega^k$. Clearly, \textbf{SYM} holds as 
	$
	\Erw[\PSI(\sigma)]=\psi_1(\sigma)/2 + \psi_{-1}(\sigma)/2 = (1\pm (1-2\eta))/2+ (1\mp (1-2\eta))/2 = 1.
	$ To see \textbf{POS} write $\PSI= \psi_{\vec s}$ for a uniform choice of $\vec s\in\Omega$. Writing $\MU_1,\MU_2,\ldots$ for independent samples from $\pi\in\PP$ and $\MU_1',\MU_2',\ldots$ for independent samples from $\pi'\in\PP$, for every $l\ge 1$ and $i\in[k]$ we obtain 
	\begin{align*}\textstyle
	\left(1-\sum_{\sigma\in\Omega^k}\PSI(\sigma)\prod_{j=1}\MU_j(\sigma_j)\right)^l &= ((1-2	\eta)\vec s)^l\prod_{j=1}^k\left(\sum_{\sigma\in\Omega}\sigma \MU_j(\sigma)\right)^l,\\\textstyle
	\left(1-\sum_{\sigma\in\Omega^k}\PSI(\sigma)\prod_{j=1}\MU'_j(\sigma_j)\right)^l &= ((1-2	\eta)\vec s)^l\prod_{j=1}^k\left(\sum_{\sigma\in\Omega}\sigma \MU'_j(\sigma)\right)^l,\\\textstyle
		\left(1-\sum_{\sigma\in\Omega^k}\PSI(\sigma)\MU_i(\sigma_I)\prod_{j\neq i}\MU'_j(\sigma_j)\right)^l & = ((1-2	\eta)\vec s)^l\left(\sum_{\sigma\in\Omega}\sigma \MU_i(\sigma)\right)^l\prod_{j\neq i}\left(\sum_{\sigma\in\Omega}\sigma \MU'_j(\sigma)\right)^l.
	\end{align*}
	Setting $X_l=\Erw[(\sum_{\sigma}\sigma\MU_1(\sigma))^l]$, $Y_l = \Erw[(\sum_{\sigma}\sigma\MU'_1(\sigma))^l]$ and writing $\Lambda$ in a logarithmic series expansion, it is clearly sufficient to show that for any $l\ge 1$ we have $\Erw[((1-2	\eta)\vec s)^l](X_l^k+(k-1)Y_l^k-kX_lY_l^{k-1})\ge 0.$
	The case of odd $l$ is immediate. For even $l$, $X$ and $Y$ are non-negative and thus $X_l^k+(k-1)Y_l^k-kX_lY_l^{k-1}$ is non-negative.
\end{proof}

\section{The lower bound}
\label{sec:lowerbound}
In this section we perform the technical computations for the Aizenman-Sims-Starr scheme. For the remainder of this section we let $\alpha\in(0,1), \beta>0$ and $D$ be arbitrary but fixed. Let $\G_n^*$ be the random graph obtained from the experiment \textbf{TCH} on $n$ vertices and let $\G_{n+1}^*$ be chosen from \textbf{TCH} on $n+1$ vertices. Moreover, we let $\G^*_{T,n}$ and $\G^*_{T,n+1}$ signify the respective graphs obtained after performing the pinning procedure from Definition \ref{def:pinning}.

Establishing the following proposition, by way of a generalization of an argument from \cite{CKPZ}, we will pave the way to proving our lower bound. 
\begin{proposition}
\label{prop:diffbound}
Let $\Delta_T(n)=\Erw[\ln Z(\G^*_{T,n+1})]-\Erw[\ln Z(\G^*_{T,{n}})]$. Then
\begin{align}
\label{eq:prop_diffbound}
\limsup_{T\to\infty}\limsup_{n\to\infty}\Delta_T(n)\le \sup_{\pi\in\cP^2_*(\Omega)}\cB(D,\pi).
\end{align}
\end{proposition}

Let us write $f(n,T)=o_T(1)$ if $\lim_{T \to \infty} \limsup_{n \to \infty}|f(n,T)|=0$. Proposition \ref{prop:diffbound} then immediately implies
\begin{align*}
\frac{1}{n}\Erw[\ln Z(\G^*_{T,n})] = \frac{1}{n}\Erw[\ln Z(\G^*_{T,1})]+\frac{1}{n}\sum_{N=1}^{n-1}\Delta_T(N)\le  o_T(1) + \sup_{\pi\in\cP^2_*(\Omega)}\cB(D,\pi).
\end{align*}
Taking the $\limsup$ over $T\to\infty$ and subsequently over $\alpha,\beta\to 0$, we obtain Proposition \ref{prop:lower_bound}. 

\subsection{The coupling}To prove Proposition \ref{prop:diffbound} we construct a coupling of the two graphs in play by first sampling maximal common subgraph $\G^0$  from which we then obtain graphs $\G'$ and $\G''$ to mirror $\G^*_{T,n}$ and $\G^*_{T,{n+1}}$ respectively. For $T>0$ let $\vec \theta$ be chosen uniformly from $[0,T]$. Choose a random $D$-partition $\vec d$ of $[n+1]$ and set $\vec \delta_1 =\vec d$ on the subdomain $[n+1]$ and $\vec \delta_1(n+1)=0$. Denote by $\vec D$ a random sample from the distribution $D$.

Let 
\begin{align*}
\lambda_0 &=\left(\min\left\{\frac{1-\alpha}{k}(n+1)\Erw D-\frac{1-\alpha}{k}\vec D, \frac{1-\alpha}{k}n\Erw D\right\}\right)_+,\\
\lambda'  &= \frac{1-\alpha}{k}n\Erw D-\lambda_0,\qquad
\lambda'' = \frac{1-\alpha}{k}(n+1)\Erw D - \lambda_0,\\
\vec \sm &= \lfloor \beta^{-1}\lambda_0 \rfloor,
\quad\sm' = \lfloor \beta^{-1}\lambda' \rfloor,
\quad\sm'' = \lfloor \beta^{-1}\lambda''\rfloor,
\quad\vec s_b = \lfloor \beta^{-1} \vec D \rfloor.
\end{align*}
Moreover, let $\vec\gamma_1,\vec\gamma_2,\ldots,\vec\gamma_{\vec \sm}$, $\vec\gamma_1',\vec\gamma_2',\ldots,\vec\gamma'_{\sm'}$, $\vec\gamma_1'',\vec\gamma_2'',\ldots,\vec\gamma''_{\sm''}$ and $\vec\gamma_{b,1},\ldots,\vec\gamma_{b,\vec s_b}$ be independently chosen from $\Po(\beta)$.
\begin{description}
\item[CPL1] On $n$ variables choose $\SIGMA_n^*\in\Omega^n$ uniformly and obtain a weighted factor graph by performing the following. Setting $s=1$ and beginning with the empty graph consisting only of variable nodes
\begin{enumerate}
\item for $i=1,\ldots,\vec\gamma_s$ add a factor node $\vec a_{s,i}$ with neighborhood chosen from 
\begin{align}\label{eq:TCHprobN}\pr[\partial \vec a_{s,i}=(x_1,\ldots,x_k)]=\frac{\prod_{j=1}^k\vec\delta_s(x_j)}{\sum_{y_1,\ldots,y_k}\prod_{j=1}^k\vec\delta_s(y_j)}
\end{align}
and weight $\vec\psi_{\vec a_{s,i}}$ chosen from
\begin{align}\label{eq:TCHprobW}\pr[\vec\psi_{\vec a_{s,i}}=\psi]=\frac{p(\psi)\psi(\SIGMA_n^*(\partial\vec a_{s,i} ))}{\Erw_{\vec \psi}[\vec\psi(\SIGMA_n^*(\partial\vec a_{s,i}))]}.
\end{align}
\item Set $\vec \delta_{s+1}=(\vec\delta_s - \nabla_s)_+$, where $\nabla_s(x)$ counts the number of times $x$ was drawn as a neighbor in round $s$. Increase $s=s+1$ and abort when $s>\vec\sm$.
\end{enumerate}
We obtain a graph that consists of $\sum_{s=1}^{\vec \sm} \Po(\beta)$ check nodes 
\item[CPL2] With probability $\vec\theta/(n+1)$ independently pin each variable node $x$ to $\SIGMA_n^*$, i.e., to each vertex $x\in [n]$ attach a unary constraint node $\psi_{a_x}$ with $\psi_{a_x}(\tau)=\vecone\{\SIGMA_n^*(x)=\tau\}$.
\end{description}
We identify the random graph resulting from \textbf{CPL1} and \textbf{CPL2} with $\G^0$.
\begin{description}
\item[CPL1'] With the pair $(\SIGMA_n^*,\G^0)$ from the former experiment perform another $\sm'$ rounds, that is for $s=\sm+1,\sm+2,\ldots,\sm+\sm'$ and $i=1,\ldots,\vec\gamma'_s$ add new weighted factor nodes $\vec a'_{s,i}$ from  \eqref{eq:TCHprobN} and \eqref{eq:TCHprobW} while updating $\vec \delta_s$ with each increment of $s$ as in \textbf{CPL1} (2).
\item[CPL2'] Then independently pin each $x\in[n]$ that is not yet pinned with probability $\vec\theta/(n(n+1-\vec \theta))$.
\end{description}
Identify the resulting graph with $\G'$.

\begin{description}
\item[CPL1'']  With $(\SIGMA_n^*,\G^0)$ and $\vec\delta_{\sm}$ from the first experiment \textbf{CPL1, CPL2} extend $\SIGMA_n^*$ to $\SIGMA_{n+1}^*$ by choosing $\SIGMA_{n+1}^*(x_{n+1})$ independently and uniformly at random. Then --- independently from the experiment \textbf{CPL1'},\textbf{CPL2'} --- from $\G^0$ create a graph $\G''$ by adding factor nodes as follows. Let $s=\sm$.
\begin{enumerate}
\item For $i=1,\ldots,\vec\gamma''_s$ add a factor node $\vec c_{s,i}$ with neighborhoods and weights from
\begin{align}
\label{eq:varn+1_1}
\pr[\vec \psi_{\vec c_{s,i}}=\psi, \partial \vec c_{s,i}=x_{i_1},\ldots,x_{i_k}]  \propto p(\psi)\psi(\SIGMA_n^*(x_{i_1},\ldots,x_{i_k}))\prod_{j=1}^k\vec \delta_s(x_{i_j})
\end{align}
\item Set $\vec \delta_{s+1}=(\vec\delta_s - \nabla_s)_+$, where $\nabla_s(x)$ denotes the number of times $x$ was drawn as a neighbor in round $s$. Halt when $s>\vec\sm+\sm''-\beta^{-1}\vec D$.
\end{enumerate}
Consequently set $\vec\delta_1=\vec \delta_s$ and reset $s=1$. Then
\begin{enumerate}[(a)]
\item For $i=1,\ldots,\vec\gamma_{b,s}$ add a factor node $\vec b_{s,i}$ with neighborhoods and weights from
\begin{align}
\label{eq:varn+1_2}
&\pr[\vec \psi_{\vec b_{s,i}}=\psi, \partial \vec b_{s,i}=x_{i_1},\ldots,x_{i_k}]\nonumber \\ \propto & \vecone\{n+1\in\{i_1,\ldots,i_k\}\}p(\psi)\psi(\SIGMA_{n+1}^*(x_{i_1},\ldots,x_{i_k}))\prod_{i_j\neq n+1}\vec \delta_s(x_{i_j})
\end{align}
\item Set $\vec \delta_{s+1}=(\vec\delta_s - \nabla_s)_+$, where $\nabla_s(x)$ denotes the number of times $x$ was drawn as a neighbor in round $s$. Halt when $s>\vec s_b$.
\end{enumerate}
\item[CPL2''] Finally, pin $x_{n+1}$ to $\SIGMA_{n+1}^*(x_{n+1})$ with probability $\vec\theta/(n+1)$.
\end{description}
Identify the resulting graph with $\G''$.

Note that by Lemma \ref{lem:nishimori} we can replace \eqref{eq:varn+1_1} and \eqref{eq:varn+1_2} by first drawing the neighborhoods from $\vec\nu_s^{\otimes k}$ and $\vec \nu_s$ respectively, then subsequently adding weights proportional to its evaluation \eqref{eq:TCHprobW}.

\begin{lemma}
	\label{lem:correctcoupling}
For sufficiently large $n$ we have 
\begin{align*}\G'  \stackrel{d}{=} \G^*_{T,n}\text{ and }\G'' \stackrel{d}{=}\G^*_{T,n+1}.
\end{align*}
\end{lemma}
\begin{proof}
Clearly, for large enough $n$ we have
\begin{align*}
\lambda_0 = \min\left\{\frac{1-\alpha}{k}(n+1)\Erw D-\frac{1-\alpha}{k}\vec D, \frac{1-\alpha}{k}n\Erw D\right\}.
\end{align*}
The choice of neighborhoods is independent from the pinning process. Hence, we might as well switch the order and perform \textbf{CPL1}, \textbf{CPL2}, \textbf{CPL1'}, \textbf{CPL2'}. In each of the experiments we add $\Po(\beta)$ many factor nodes from the distribution \eqref{eq:TCHprobN}. The neighborhoods are then chosen independently from everything but the planted coloring $\SIGMA^*$ and the neighborhood. Consequently, it suffices to compare the processes that generate these random neighborhoods. Because the number of rounds that we perform in both \textbf{CPL1}, \textbf{CPL2} and the model \textbf{TCH} on $n$ vertices are $\lfloor k^{-1}(1-\alpha)n\Erw D\rfloor$ we can perfectly couple the occurences of each round. Moreover, in the process of performing \textbf{CPL1'} and then \textbf{CPL2'} each variable gets pinned independently with probability $\vec\theta/n$ and therefore the procedure results in the graph $\G_{T,n}^*$.

It remains to prove the second distributional equality. On the event that $\vec D = 0$ we see in the same fashion as with $\G'$ that the procedure \textbf{CPL1}, \textbf{CPL1''} yields $\G^*_{n+1}$. Suppose that $\vec D>0$. Even though we first consider $x_{n+1}$ in the last step, our choice of $\vec s_b$ and \eqref{eq:varn+1_1} ensure that its degree in expectation over the outcome of our Poisson random variables $\vec \gamma_{b,1},\ldots,\vec\gamma_{b,\vec s b}$ is as if chosen from \eqref{eq:TCHprobN} on $n+1$ vertices. By \eqref{eq:TCHprobN}, \eqref{eq:varn+1_1} and \eqref{eq:varn+1_2} so is the neighborhood distribution. Finally, as with $\G'$, due to \textbf{CPL2}, \textbf{CPL2''} each variable node in $G''$ is independently pinned with probability $\vec\theta/(n+1)$ and therefore $\G''$ is distributed as $G_{T,n+1}$.
\end{proof}

By Lemma \ref{lem:correctcoupling} we now have a characterization of the left-hand side in \eqref{eq:prop_diffbound}
\begin{align}\label{eq:DeltaEqn}
\Delta = \Erw\left[\ln \frac{Z(\G'')}{Z(\G')}\right] = \Erw\left[\ln \frac{Z(\G'')}{Z(\G^0)}\right]-\Erw\left[\ln \frac{Z(\G')}{Z(\G^0)}\right]
\end{align}
which is approachable by calculating the contributions that are added to the partition function when going from $\G^0$ to either $\G'$ or $\G''$. To assure that \eqref{eq:DeltaEqn} mediates the assertion in \eqref{eq:prop_diffbound} we establish the following fact.

\begin{lemma}
\label{lem:G0GTast}
We have $\|\G^0-\G_T^*\|=o(1)$.
\end{lemma}
\begin{proof}$T$ is fixed and $\vec\theta$ is chosen uniformly in $[0,T]$. If $\vec U^0$ denotes the random set that is pinned in $\G^0$ then each variable in $\vec U^0$ was chosen with probability $\vec\theta$. As $T/n - T/(n+1) = o(1)$ the sets $\vec U$ and $\vec U^{0}$ can be coupled such that they coincide with probability $1-o(1)$. Moreover, $
\Erw[\vec |E(\G_T^*)|-|\vec m(\G^0)|]=O(1)
$
whereas the variance of $|E(\G_T^*)|$, $|E(\G^0)|$ is at least $\Omega(n)$. Hence, all but $O(1)$ factor nodes can be optimally coupled, such that the neighborhoods are at most $o(1)$ apart in total variation distance.
\end{proof}

Before we proceed to calculate \eqref{eq:DeltaEqn} we want to capture that our coupling is typically well-behaved.
\begin{claim}
\label{claim:bethe1}
Let $\cU$ be the event that given $\G^0$ the process \textbf{CPL2'} does not pin any additional vertices. Then
\begin{align}
\label{eq:Ubethe1}
\Erw\left[\ln({Z(\G')}/{Z(\G^0)})|\G^0\right]=o_T(1)+\Erw\left[\vecone_{\cU}\ln({Z(\G')}/{Z(\G^0)})|\G^0\right].
\end{align}
Let $\vec Y$ be the set of vertices belonging to the neighborhood chosen in \textbf{CPL1'}. Let $\vec m_0$ denote the number of factor nodes added during \textbf{CPL1'}. There is $\eps=o(1)$ such that given $T$ is sufficiently large and given $\G^0$ the event 
\begin{align*}
\cY=\{ \|\mu_{\G^0,\vec Y}-\otimes_{y\in \vec Y}\mu_{\G^0,y}\|_{TV}\le \eps \text{ and } |\vec Y|=k \vec m_0\}
\end{align*}
occurs with probability at least $1-\eps$.
\end{claim}

\begin{proof}
To validate the first part, observe that for any graph $\G^0$ during \textbf{CPL2'} each variable node is independently pinned with probability $\vec\theta/(n(n-1+\vec\theta))\le T/(n(n-1))$. Thus $\pr[\cU|\G^0]\ge 1-T/(n-1)\ge 1-2T/n$. Equation \eqref{eq:Ubethe1} is immediate as all weights added during the process are strictly positive.

Each neighborhood $(x_1,\ldots,x_k)$ in \textbf{CPL1'} is chosen from a product measure proportional to a slight perturbation of $\prod_{i=1}^k\vec\delta_{\sm}(x_i)$, where the neighbors are chosen from a set of size $O(\alpha n)$ \whp~ The expected number of new factor nodes is $\Erw[\vec m_0] =O(\beta)$, whence $\pr[|\vec Y|=k\vec m_0|\G^0]=1-o(1)$. Moreover, by the first part of Lemma \ref{lem:epsSymmetry} there is $\eps=o(1)$ such that for sufficiently large $T$ the graph $\G^*_T$ is $\eps$-symmetric with probability at least $1-\eps$.  Combining this with the second part of Lemma \ref{lem:epsSymmetry} we find $\eps'=o(1)$ such that given $\G^0$ the event $\cY=\{ \|\mu_{\G^0,Y}-\otimes_{y\in Y}\mu_{\G^0,y}\|_{TV}\le \eps \text{ and } |Y|=k \vec m_0\}$ occurs with probability at least $1-\eps'$ as long as $T>T_0(\eps')$ is sufficiently large.
\end{proof}

\begin{lemma}
\label{lem:bethe1}
Let $\pi_{\G^0}$ be the empirical distribution of the Gibbs marginals of $\G^0$. We have
\begin{align*}
\Erw\left[\ln({Z(\G')}/{Z(\G^0)})|\G^0\right]=o_T(1)+(1-\alpha)\frac{k-1}{k\xi}\Erw\left[\vec D\Lambda\left(\sum_{\tau\in \Omega^k}\vec\psi(\tau)\prod_{h=1}^k \vec\mu_h^{(\pi_{\G^0})}(\tau_j)\right)\right]
\end{align*}
with probability at least $1-o_T(1)$ over the choice of $\G^0$.
\end{lemma}
\begin{proof}
On the event $\cU\cap \cY$ the graph $\G'$ is obtained from $\G^0$ by simply adding $\vec m'=Po(\lambda')$ check nodes $b_1,\ldots, b_{\vec m'}$ with neighborhoods and weights chosen from \eqref{eq:TCHprobN} and \eqref{eq:TCHprobW} respectively.
\begin{align}
\label{eq:G1_1}
\Erw\left[\ln({Z(\G')}/{Z(\G^0)})|\G^0,\SIGMA_n^*\right]&=o_T(1)+\Erw\left[\ln\left\langle\prod_{i=1}^{\vec m'}\vec\psi_{b_i}(\SIGMA(\partial_1 b_i),\ldots\SIGMA(\partial_k b_i))\right\rangle_{\G^0}\bigg|\G^0\right]\nonumber\\
&=o_T(1)+\Erw\left[\ln\sum_{\tau\in\Omega^Y}\mu_{\G^0,Y}(\tau)\prod_{i=1}^{\vec m'}\vec\psi_{b_i}(\tau(\partial_1 b_i),\ldots,\tau(\partial_k b_i))\bigg|\G^0\right].
\end{align}
As the factor nodes $b_1,\ldots,b_{\vec m'}$ are independently chosen from the same distribution \whp, with Claim \ref{claim:bethe1} the expression \eqref{eq:G1_1} simplifies to
\begin{align}
\label{eq:G1_2}
&o_T(1)+\Erw\left[\sum_{i=1}^{\vec m'}\ln\sum_{\tau\in\Omega^k}\vec\psi_{b_i}(\tau)\prod_{h=1}^k\mu_{\G^0,\partial_h b_i}(\tau_h)\bigg|\G^0\right]\nonumber\\
=&o_T(1)+ \frac{(1-\alpha)(k-1)}{k}\Erw\left[\vec D\ln\sum_{\tau\in\Omega^k}\vec\psi_{b_1}(\tau)\prod_{h=1}^k\mu_{\G^0,\partial_h b_1}(\tau_h)\bigg|\G^0\right].
\end{align}
Hence, with $\vec i_1,\ldots,\vec i_k$ chosen independently and uniformly at random from $[n]$, using \textbf{SYM2} in the distribution \eqref{eq:TCHprobN},\eqref{eq:TCHprobW} of $b_1$ and writing \eqref{eq:G1_2} with the empirical distribution of $\G^0$ we obtain
\begin{align*}
\Erw\left[\ln({Z(\G')}/{Z(\G^0)})|\G^0\right]=o_T(1)+\frac{(1-\alpha)(k-1)}{k\xi}\Erw\left[\vec D\Lambda\left(\sum_{\tau\in \Omega^k}\vec\psi(\tau)\prod_{h=1}^k \vec\mu_h^{(\pi_{\G^0})}(\tau_j)\right)\bigg|\G^0\right].
\end{align*}
\end{proof}

\begin{claim}
\label{claim:bethe2}
Let $\cU'$ be the event that given $\G^0$ the process \textbf{CPL1''} yields $\vec m'=0$ factor nodes and \textbf{CPL2''} does not pin $x_{n+1}$. If $\beta$ is sufficiently small, then 
\begin{align*}
\Erw\left[\ln({Z(\G'')}/{Z(\G^0)})|\G^0\right]=o_T(1)+\Erw\left[\vecone_{\cU'}\ln({Z(\G'')}/{Z(\G^0)})|\G^0\right].
\end{align*}
For $j=1,\ldots,\vec\gamma_{b,\vec s_b}$ let $\vec h_j$ be independently uniformly chosen indices in $[k]$ and $\vec y_j=(\vec y_{j1,},\ldots,\vec y_{j,k})$ uniformly random choices of neighborhoods subject to the condition that $\vec y_{j,\vec h_j}=x_{n+1}\neq \vec y_{j,i}$, $i\neq \vec h_j$.  Let $\vec Y'$ be the set of neighbors of $b_1,\ldots,b_{\vec s_b}$ chosen in \textbf{CPL1''} without $x_{n+1}$. Let $\vec Y=\{\vec y_{j,h}:j\le\vec\gamma_{b,\vec s_b},h\in[k]\}\setminus\{x_{n+1}\}$ and $\vec D$ be a sample from $D$. Then $\vec Y$ and $\vec Y'$ are mutually contiguous and there is $\eps=o(1)$ such that given $\G^0$ the event 
\begin{align*}
\cY'=\{ \|\mu_{\G^0,Y'}-\otimes_{y\in Y'}\mu_{\G^0,y}\|_{TV}\le \eps \text{ and } |\vec Y'|=(k-1) \vec D\}
\end{align*}
occurs with probability at least $1-\eps$.
\end{claim}
\begin{proof}
The probability of $x_{n+1}$ not being pinned is at least $1-T/(n+1)=o(1)$. As $D$ has finite support, if $\beta$ is sufficiently small we have $\Erw[\vec m']=0$ and thus the first assertion follows.

Given $\G^0$, for each $j=1,\ldots,\vec\gamma_{b,\vec s_b}$ the probability of drawing a neighbor $x$ with $\vec\delta_s(x)=0$ is upper bounded by $n^{-1}\alpha^{-1}k$, while $\Erw[\vec\gamma_{b,\vec s_b}]=(1-\alpha)k^{-1}\Erw[D]=O(1)$. Therefore $\vec Y'$ is of size $(k-1)\vec D$ asymptotically almost surely. Moreover, because $\vec\delta_s(\vec y_{j,i})$ is bounded by a small constant for any $1\le j\le \vec D$, $i\le k$, we have mutually contiguity of $\vec Y$ and $\vec Y'$. As $\G_T^*$ is $\eps$-symmetric for sufficiently large $T$, we may again apply Lemma \ref{lem:epsSymmetry} to obtain a sequence $\eps'=o(1)$ such that
$
\cY'=\{ \|\mu_{\G^0,Y'}-\otimes_{y\in Y'}\mu_{\G^0,y}\|_{TV}\le \eps \text{ and } |\vec Y'|=(k-1) \vec D\}
$
occurs with probability at least $\eps'$.
\end{proof}

\begin{lemma}
\label{lem:bethe2}
Let $\pi_{\G^0}$ be the empirical distribution of the Gibbs marginals of $\G^0$. We have
\begin{align*}
\Erw\left[\ln({Z(\G'')}/{Z(\G^0)})|\G^0\right]=o_T(1)+\Erw\left[\frac{\xi^{-\vec D}}{|\Omega|}\Lambda\left(\sum_{\sigma\in\Omega}\prod_{i=1}^ {\vec D}\sum_{\tau\in\Omega^k}\vecone\{\tau_{\vec h_i}=\sigma\}\vec\psi_i(\tau)\prod_{j\neq \vec h_i}\vec\mu_{ki+j}^{(\pi_{\G^0})}(\tau_j)\right)\right]
\end{align*}
with probability at least $1-o_T(1)$ over the choice of $\G^0$.
\end{lemma}
\begin{proof}
As in the proof of the previous lemma, given that $\cU'\cap \cY'$ occurs, Claim \ref{claim:bethe2} implies that $\G''$ is obtained from $\G^0$ by simply adding $\vec D$ weighted check nodes from \eqref{eq:varn+1_2} and therefore
\begin{align}
\label{eq:G2_1}
\Erw\left[\ln\frac{Z(\G'')}{Z(\G^0)}\bigg|\G^0\right]=o_T(1)+\Erw\left[\ln\sum_{\tau\in \Omega^{Y\cup\{x_{n+1}\}}}\mu_{G^0,Y}(\tau|_Y)\prod_{i=1}^{\vec D}\vec\psi_{b_i}(\tau(\partial_1 b_i),\ldots,\tau(\partial_k b_i))\bigg|\G^0\right].
\end{align}
Hence, with the natural extension of $\SIGMA_n^*$ to $\SIGMA_{n+1}^*$ and $\vec h_1,\vec h_2,\ldots$ uniformly chosen from $[k]$, for $i=1,\ldots,D$  let $(\vec \omega_{i,1},\ldots,\vec\omega_{i,k})\in\Omega^k$ and $\vec\psi'_i$ be chosen from
\begin{align*}
\pr[(\vec \omega_{i,1},\ldots,\vec\omega_{i,k})=(\omega_1,\ldots,\omega_k),\vec\psi' = \psi] \propto\vecone\{\omega_{i,\vec h_i}=\SIGMA_{n+1}^*(x_{n+1})\}\xi^{-1}p(\psi)\psi(\omega_1,\ldots,\omega_k).
\end{align*}
Together with the $\eps$-symmetry statement of Claim \ref{claim:bethe2} equation \eqref{eq:G2_1} becomes 
\begin{align*}
o_T(1)+\Erw\left[\ln\sum_{\sigma\in\Omega}\prod_{i=1}^ {\vec D}\sum_{\tau\in\Omega^k}\vecone\{\tau_{\vec h_i}=\sigma\}\vec\psi'_i(\tau)\prod_{h\neq\vec h_i}\mu_{\G^0,\vec y_{i,h}}(\tau_h)|\G^0\right].
\end{align*}
Again, writing out the probabilities for the independently chosen check nodes, with \textbf{SYM} we get
\begin{align*}
o_T(1)+\Erw\left[\frac{\xi^{-\vec D}}{|\Omega|}\Lambda\left(\sum_{\sigma\in\Omega}\prod_{i=1}^ {\vec D}\sum_{\tau\in\Omega^k}\vecone\{\tau_{\vec h_i}=\sigma\}\vec\psi_i(\tau)\prod_{h\neq\vec h_i}\vec\mu_{ki+h}^{(\pi_{\G^0})}(\tau_h)\right)\right].
\end{align*}
\end{proof}

Finally, to prove Proposition \ref{prop:diffbound} we are going to make use of the following fact that is immediate from \ref{lem:nishimori}.
\begin{fact}[{\cite[Corollary 3.13]{CKPZ}}]
\label{fact:almostunif}
For all $T\ge0$ and all $\omega\in\Omega$ we have
\begin{align*}
\Erw\langle ||\SIGMA^{-1}(\omega)|-n/|\Omega||\rangle_{\G^*_T}=o(1).
\end{align*}
\end{fact}

Moreover, note that by performing a continuous transformation of the functional utilized in \cite{CKPZ} we preserve the following property.

\begin{fact}[{\cite[Lemma 2.9]{CKPZ}}]
\label{fact:continuous}
The functional $\pi\in\cP^2(\Omega)\to\cB(D,\pi)$ is weakly continuous.
\end{fact}

\begin{proof}[Proof of Proposition \ref{prop:diffbound}]
By Lemmata \ref{lem:bethe1} and \ref{lem:bethe2} we have established that
\begin{align}
\label{eq:diffbound_G0}
\Delta_T(n)\le o_T(1)+\Erw[\cB(D,\pi_{\G^0})].
\end{align}
To bridge the gap to $\pi\in\cP^2_*(\Omega)$ we have to carve out that $\pi_{\G^0}$ is arbitrarily close to the set $\cP^2_\ast(\Omega)$ as $n\to\infty$. Fact \ref{fact:continuous} together with \eqref{eq:diffbound_G0} then yields the assertion when taking the limits in the specified order.

Because $\SIGMA^*$ is the uniform distribution on $\Omega^n$, Fact \ref{fact:almostunif} gives $\Erw\langle ||\SIGMA^{-1}(\omega)|-n/|\Omega||\rangle_{\G^*_T}=o(1)$. Therefore, using Lemma \ref{lem:G0GTast} we obtain
\begin{align*}
&\Erw\left|\int \mu(\omega)\mathrm{d}\pi_{\G^0}(\mu)-|\Omega|^{-1}\right|=\Erw\left|\frac{1}{n}\sum_{x\in [n]}\langle \vecone\{\SIGMA(x)=\omega\}\rangle_{\G^0}-|\Omega|^{-1}\right|\\
\le&\Erw\left\langle\frac{|\SIGMA^{-1}(\omega)|}{n}-\frac{1}{|\Omega|}|\right\rangle_{\G^0}=o(1),
\end{align*}
whence in expectation over $\G^0$ the measure $\int\mu \mathrm{d}\pi_{\G^0}(\mu)$ converges to the uniform distribution in total variation distance. That is \whp~ there is $\alpha(\G^0)\ge 0$, $\Erw[\alpha(\G^0)]=o(1)$ and a measure $\nu(\G^0)\in\cP(\Omega)$ such that the convex combination $(1-\alpha(\G^0))\pi_{\G^0}+\alpha(\G^0)\delta_{\nu(\G^0)}\in\cP^2_*(\Omega)$ closes the gap.
\end{proof}

\section{The upper bound}
\label{sec:upper_bound}
In this section we carry out the calculations for the Guerra-Toninelli interpolation. For any given $\pi\in\PP$ we set up a family of graphs $(\G_t)_{t\in[0,1]}$, by which we can interpolate between the original graph $\G_1=\G^*$ and a graph $\G_0$ with free energy $-n\cB(D,\pi)+o(n)$. By proving that the derivative $\partial/\partial t\Erw[Z(\G_t)]$ is positive on the entire interval, we obtain
\begin{align}
\label{eq:upbound_bethe}
-\frac{1}{n}\Erw[\ln Z(\G^*)]\le -\cB(D,\pi)+o(1).
\end{align}
Throughout the section we assume that $\alpha\in(0,1),\beta>0$, the degree distribution $D$ and $\pi\in\PP$ are arbitrary but fixed. 

\subsection{The interpolation}
In the interpolation we will utilize the fact, that our graph $\G=\G_{\alpha,\beta,D}$ consists of $\sm$ layers of Poissonian degree. This being the case, for any $s\in\{1,\ldots,\sm\}$ we define an interpolation as follows. A codeword in the interpolation model consists of $s-1$ layers of parity checks, followed by a layer, where with probability $1-t'$, $t'\in[0,1]$ each parity check is replaced by a repetition of the codebits it contains, and a final $\sm-s$ layers of simple blocks of repetition code altogether satisfying the degree distribution. To make this precise, we define a random factor graph model $\G_{s,t}=\G_{s,t}(\vec m,\vec \gamma)$ as follows.
\begin{description}
\item[I1] Draw a random $D$-partition $\vec d$ of $[n]$. Let $\vec X'$ be chosen from $\Po(\beta t)$, $\vec X''$ be chosen from $\Po(\beta(1-t))$ and let $\vec X_1,\vec X_2,\ldots$ be a sequence of $\Po(\beta)$ random variables all mutually independent. Define vectors $\vec m$ and $\vec \gamma$ by letting
\begin{align*}
\vec m_\ell = \begin{cases}X_\ell,&\ell<s\\X',&\ell=s\\0,&\ell>s\end{cases}, \qquad\vec \gamma_s = \begin{cases}0,&\ell<s\\X'',&\ell=s\\X_\ell,&\ell>s\end{cases}\qquad\text{ for }\ell=1,\ldots,\sm.
\end{align*}
\item[I2] For $i=1,\ldots,\vec m_s$ add a $k$-ary factor node $\vec a_{s,i}$ with neighborhood from
\begin{align*}
\pr[\partial \vec a_{s,i}=(x_1,\ldots,x_k)]=\frac{\prod_{j=1}^k\vec\delta_s(x_j)}{\sum_{y_1,\ldots,y_k}\prod_{j=1}^k\vec\delta_s(y_j)}.
\end{align*}
For $j=1,\ldots,k\vec\gamma_s$ add a unary factor $\vec b_{s,j}$ with neighbors chosen from
\begin{align}
\label{eq:unaryN}
\pr[\partial \vec b_{s,j}=x]=\frac{\vec\delta_s(x)}{\sum_{y}\vec\delta_s(y)}.
\end{align}
 Set $\vec \delta_{s+1}=(\vec\delta_s - \nabla_s)_+$, where $\nabla_s(x)$ denotes the number of times $x$ was drawn as a neighbor in round $s$. Increase $s$ and abort when $s>\vec\sm$.
\item[I3]To each $k$-ary factor node $\vec a$ in the graph independently assign a weight function $\vec\psi_{\vec a}$ chosen from $p$.
\item[I4] To each unary factor node $\vec b$ in the graph independently assign a unary weight function $\vec\psi_{\vec b}$ as follows. Choose $\PSI$ from $p$, indenpendently choose $i$ from the uniform distribution on $[k]$ and choose $\MU_1,\ldots,\MU_k$ iid from $\pi$. Let $\vec\psi_{\vec b}$ be the map \begin{align*}
 \sigma\mapsto \sum_{\tau_1,\ldots,\tau_k}\PSI(\tau_1,\ldots,\tau_{i-1},\sigma,\tau_{i+1},\ldots,\tau_k)\prod_{h\neq \vec i}\MU_h(\tau_h).
 \end{align*}
\end{description}
Having established the interpolation null-model, we can now define our original interpolation in the teacher-student model by reweighing. With a signal $\sigma:[n]\to\Omega$ we define the distribution $\G^*_{s,t}(\sigma)=\G^*_{s,t}(\sigma,\vec m,\vec \gamma)$ by letting
\begin{align*}
\pr[\G^*_{s,t}(\sigma)\in\cA]=\frac{\Erw[\psi_{\G_{s,t}}(\sigma)\vecone\{\G_{s,t}\in\cA\}]}{\Erw[\psi_{\G_{s,t}}(\sigma)]}.
\end{align*}
If the signal $\SIGMA$ is chosen uniformly at random, we write $\G^*_{s,t}=\G^*_{s,t}(\vec m,\vec \gamma)$ for $\G^*_{s,t}(\SIGMA,\vec m,\vec \gamma)$.
Notice that $s=\sm, t'=1$ yields the original graph model $\G^*$ and $s=1, t'=0$ corresponds to the graph of a simple repetition code. All we do by layer-wise interpolation is to split the interpolation interval into $\sm$ intervals of equal length. Finally, to ensure symmetry, we apply the pinning procedure. For this purpose fix $T>0$.
\begin{description}
\item[IP1] Choose $\check\SIGMA$ from the uniform distribution on $\Omega^n$.
\item[IP2] Choose a random graph $\G=([n],\vec F,(\partial a)_{a\in \vec F},(\vec\psi_{\vec a})_{\vec a\in F})$ from $\G^*_{s,t}(\check\SIGMA,\vec m,\vec \gamma)$.
\item[IP3] Choose $\vec\theta$ in $[0,T]$ uniformly at random and let $\vec U$ be a random $\vec\theta/n$-subset of $[n]$.
\item[IP4] To each $x\in \vec U$ connect a unary check node $a_x$ with weight function $\psi_{a_x}(\tau)=\vecone\{\check\SIGMA(x)=\tau\}$.
\end{description}
Let us write $\G^*_{T,s,t}$ for the resulting graph. Note that by Lemma \ref{lem:epsSymmetry} there is $T_0$ that only depends on $\eps$ and $\Omega$ such that for $T\ge T_0$ the Gibbs measure of $\G^*_{T,s,t}$ is $\eps$-symmetric with probability at least $1-\eps$. Thus, we can fix a sufficiently large $T>0$ before performing the interpolation and guarantee that throughout the process the family of graphs remains $\eps$-symmetric.

The tally of total factor nodes in the partition function during interpolation is accounted for in a correction term. To this end, with independent samples $\vec\mu_1,\vec\mu_2,\ldots$ from $\pi$, let

\begin{align*}
\Gamma_{s,t}=\frac{(s+t-1)\beta(k-1)}{\xi}\Erw\left[\Lambda\left(\sum_{\tau\in\Omega^k}\vec\psi(\tau)\prod_{j=1}^k\vec\mu_j^{(\pi)}(\tau_j)\right)\right].
\end{align*}

The following lower bound is the main ingredient to our interpolation argument as it tethers the derivative arbitrarily closely to zero.
\begin{proposition}
\label{prop:derivBound}
Let \begin{align*}
\Phi_{T,s}:t\in[0,1] \mapsto \left(\Erw[\ln Z(\G^*_{T,s,t})]+\Gamma_{s,t}\right)/n.
\end{align*}
Then for all $t\in[0,1]$ and all $s\in[\sm]$ we have $\Phi'_{T,s}(t)>o_T(1)n^{-1}$.
\end{proposition}

We will prove Proposition \ref{prop:derivBound} in section \ref{sec:proofPropDeriv}. Let us first see how Proposition \ref{prop:derivBound} implies \eqref{eq:upbound_bethe} and as such Proposition \ref{prop:upper_bound}.
\subsection{Proof of Proposition \ref{prop:upper_bound}}
 With the fundamental theorem of calculus we write
\begin{align}
\label{eq:theinterpolation}
&\frac{1}{n}\Erw[\ln Z(\G_{T,n}^*)] + \frac{1}{n}\Gamma_{\sm,1}= \frac{1}{n}(\Erw[\ln Z(\G_{T,\sm,1}^*)]+\Gamma_{\sm,1})\nonumber\\
=&\frac{1}{n}(\Erw[\ln Z(\G_{T,1,0}^*)]+\Gamma_{1,0})+\sum_{s=1}^{\sm} \int_0^t \Phi'_{T,s}(t)\dd t\\
\ge& \frac{1}{n}\Erw[\ln Z(\G_{T,1,0}^*)]-\frac{1}{n}\sum_{s=1}^{\sm} \int_0^t o_T(1)\dd t\nonumber\\
=&\frac{1}{n}\Erw[\ln Z(\G_{T,1,0}^*)]+o_T(1).
\end{align}
If $(\vec\psi_i)_{i\ge1}$ is a sequence with entries independently chosen from $p$, $(\vec \mu_{i,j})_{i,j\ge 1}$ has entries independently chosen from $\pi$, $\vec h_1,\vec h_2,\ldots$ are independently uniform choices from $[k]$ and $\vec\gamma$ is a random sample from $D$, a simple calculation unfolds that
\begin{align}
\label{eq:intEndBethe}
\frac{1}{n}\Erw[\ln Z(\G_{T,1,0}^*)]  = \frac{1}{|\Omega|}\Erw\left[\xi^{-\vec\gamma}\Lambda\left(\sum_{\sigma\in\Omega}\prod_{b=1}^{\vec\gamma}\sum_{\tau\in\Omega^k}\vecone\{\tau_{\vec h_b}=\sigma\}\vec\psi_b(\tau)\prod_{j\in[k]\setminus\{\vec h_b\}}\vec\mu_{b,j}(\tau_j)\right)\right].
\end{align}
Also, by definition of $\Gamma_{s,t}$ and $\sm$,
\begin{align}
\label{eq:intEndBethe2}
\Gamma_{\sm,1}\le \frac{(1-\alpha)n\Erw[D]}{k}\frac{(k-1)}{\xi}\Erw\left[\Lambda\left(\sum_{\tau\in\Omega^k}\vec\psi(\tau)\prod_{j=1}^k\vec\mu_j^{(\pi)}\right)\right]
\end{align}
Plugging \eqref{eq:intEndBethe} and \eqref{eq:intEndBethe2} into \eqref{eq:theinterpolation} and consequently taking the $\liminf$ of $T\to\infty$, we have
\begin{align}
&\liminf_{n\to\infty} \frac{1}{n}\Erw[Z(\G^*_{\alpha,\beta,D})]\nonumber\\
\ge & \liminf_{n\to\infty} \left(\frac{1}{|\Omega|}\Erw\left[\xi^{-\vec\gamma}\Lambda\left(\sum_{\sigma\in\Omega}\prod_{b=1}^{\vec\gamma}\sum_{\tau\in\Omega^k}\vecone\{\tau_{\vec h_b}=\sigma\}\vec\psi_b(\tau)\prod_{j\in[k]\setminus\{\vec h_b\}}\vec\mu_{b,j}(\tau_j)\right)\right]\right.\nonumber\\
& \left.- \frac{(1-\alpha)(k-1)}{k\xi}\Erw[D]\Erw\left[\Lambda\left(\sum_{\tau\in\Omega^k}\vec\psi(\tau)\prod_{j=1}^k\vec\mu_j^{(\pi)}\right)\right]\right).\label{eq:BetheBoundApprox}
\end{align}
The weights in $\Psi$ are strictly positive. Hence, Proposition \ref{prop:goodcoupling} allows us to compare the free energy of our exact model $Z(\G^*_{D})$ with the approximative free energy as
\begin{align*}\Erw[\ln Z(\G_D^*)]=\Erw[\ln Z(\G^*_{\alpha,\beta,D})]+O(\alpha n).
\end{align*}
Thus, \eqref{eq:BetheBoundApprox} extends to the exact model when ultimately taking our approximation to the limit
\begin{align*}
&\limsup_{n\to\infty} -\frac{1}{n}\Erw[Z(\G_D^*)]\\
 \le& -\sup_{\pi\in\PP}\limsup_{\alpha,\beta\to0}\limsup_{n\to\infty} \left(\frac{1}{|\Omega|}\Erw\left[\xi^{-\vec\gamma}\Lambda\left(\sum_{\sigma\in\Omega}\prod_{b=1}^{\vec\gamma}\sum_{\tau\in\Omega^k}\vecone\{\tau_{\vec h_b}=\sigma\}\vec\psi_b(\tau)\prod_{j\in[k]\setminus\{\vec h_b\}}\vec\mu_{b,j}(\tau_j)\right)\right]\right.\\
& \left.\qquad\qquad\qquad\qquad\qquad\qquad\qquad\qquad\qquad- \frac{(1-\alpha)(k-1)}{k\xi}\Erw[D]\Erw\left[\Lambda\left(\sum_{\tau\in\Omega^k}\vec\psi(\tau)\prod_{j=1}^k\vec\mu_j^{(\pi)}\right)\right]\right)\\
=&-\sup_{\pi\in\PP}\cB(D,\pi).
\end{align*}

\subsection{Proof of Proposition \ref{prop:derivBound}}
\label{sec:proofPropDeriv} To prove Proposition \ref{prop:derivBound} we derive a more practical expression of the derivative $\Phi'_{T,s}$ which is comparable to the expression from \textbf{POS}.

\begin{proposition}
\label{prop:derivIsXi}
Let $(\check\SIGMA,\G^*_{T,s,t})$ be chosen from the \textbf{IP} experiment. Let $\PSI$ be chosen from $p$, $\vec\mu_1,\vec\mu_2,\ldots,\vec\mu_k$ be chosen from $\pi$, all mutually independent. For $s\in[\sm]$ let $\vec \delta_s$ be chosen from \textbf{AP2'} in Definition \ref{def:abApprox} and set 
$$\vec\nu_s =  \frac{\vec\delta_s}{\sum_{y\in[n]}\vec\delta_s(y)}.$$
With $\vec y,\vec y_1,\ldots,\vec y_k$ chosen uniformly from the set of variable nodes, we let
\begin{align*}
\Xi_{s,t,l} =& \Erw\left[\prod_{i=1}^k\vec\nu_s(\vec y_i)\langle 1-\PSI(\SIGMA(\vec y_1),\ldots,\SIGMA(\vec y_k)\rangle^l_{\G^*_{T,s,t}}\right]\\
 &- \Erw\left[\vec\nu_s(\vec y)\langle 1-\sum_{\tau\in\Omega^k}\PSI(\tau)\vecone\{\tau_i = \SIGMA(\vec y_i)\}\prod_{j\neq i}\vec\mu_j(\tau_j)\rangle^l_{\G^*_{T,s,t}}\right]\\
&+(k-1)\Erw\left[\left(1-\sum_{\tau\in\Omega^k}\PSI(\tau)\prod_{j=1}^k\vec\mu_j(\tau_j)\right)^l\right].
\end{align*}
Then uniformly for all $t\in(0,1)$, $s\in[\sm]$ and $T\ge 0$ we have
\begin{align*}
\frac{\dd }{\dd t}\Phi_{T,s}(t) = o_T(1)n^{-1}+ 
\frac{\beta}{n\xi}\sum_{l\ge 2}\frac{\Xi_{s,t,l}}{l(l-1)}.
\end{align*}
\end{proposition}
To prove Proposition \ref{prop:derivIsXi} we begin by rewriting $\frac{\partial}{\partial t} \Erw[\ln Z(\G^*_{T,s,t})]$ into a similar expression. Note that the Poisson distribution with parameters $\lambda>0$ satisfies
\begin{align*}
\frac{\partial}{\partial \lambda}\Po(\lambda)(\{m\})=\frac{\partial}{\partial \lambda}\frac{\lambda^m}{m!}e^{-\lambda}=\frac{\lambda^{m-1}}{(m-1)!}e^{-\lambda}-\frac{\lambda^m}{m!}e^{-\lambda}=\Po(\lambda)(\{m-1\})-\Po(\lambda)(\{m\}),\quad m\ge 1.
\end{align*}
Hence, as $\vec m$ and $\vec \gamma$ are independent but add $\Po(\beta)$ new neighbors to each layer when put together, differentiating $\Erw[\ln Z(\G^*_{T,s,t})]$ with respect to $t$ yields 
\begin{align*}
&\frac{\partial}{\partial t}\Erw[\ln Z(\G^*_{T,s,t})]\\
=&\sum_{m,\gamma}\Erw[\ln Z(\G^*_{T,s,t})|\vec m_s = m,\vec \gamma_s = \gamma]\frac{\partial}{\partial t}\Po(t\beta)(\{m\}) \Po((1-t)\beta)(\{\gamma\})\\
=&\beta\sum_m \left(\Erw[\ln Z(\G^*_{T,s,t})|\vec m_s = m+1] -\Erw[\ln Z(\G^*_{T,s,t})|\vec m_s = m]\right)\Po(t\beta)(\{m\})\\
&-\beta\sum_\gamma\left(\Erw[\ln Z(\G^*_{T,s,t})|\vec \gamma_s = \gamma+1] -\Erw[\ln Z(\G^*_{T,s,t})|\vec \gamma_s = \gamma]\right)\Po((1-t)\beta)(\{\gamma\})\\
=&\beta\left(\Erw[\ln Z(\G^*_{T,s,t}(\vec m+\vecone_s, \vec \gamma))]-\Erw[\ln Z(\G^*_{T,s,t}(\vec m, \vec \gamma))]\right)\\
&-\beta\left(\Erw[\ln Z(\G^*_{T,s,t}(\vec m, \vec \gamma+\vecone_s))]-\Erw[\ln Z(\G^*_{T,s,t}(\vec m, \vec \gamma))]\right).
\end{align*}
The term $\beta\Erw[\ln Z(\G^*_{T,s,t}(\vec m, \vec \gamma))]$ cancels and we can write
\begin{align}
\label{eq:interpolation}
\frac{\partial}{\partial t}\Erw[\ln Z(\G^*_{T,s,t})]=\beta\Erw[\ln Z(\G^*_{T,s,t}(\vec m+\vecone_s, \vec \gamma))]-\beta\Erw[\ln Z(\G^*_{T,s,t}(\vec m, \vec \gamma+\vecone_s))]
\end{align}
Now consider a graph model $\G^{*,1}_{T,s,t}$, where we slightly alter the procedure \textbf{I2} as follows. 
\begin{description}
\item[I2'] For any $\ell\neq s$ we construct the graph as described by \textbf{I2}. If $\ell=s$, instead of increasing $s$ and moving onto \textbf{I3}, we add another $k$ unary check nodes $\vec b^{(1)}_1,\ldots,\vec b^{(1)}_k$ with  neighborhoods chosen from \eqref{eq:unaryN}, update $\vec\delta_{s+1}$ accordingly and equip each of the new check nodes with constant weight functions $\psi_{\vec b^{(h)}}=1$, $h=1,\ldots,k$. Afterwards we increase $s$ and continue. 
\end{description}
In short $\G^{*,1}_{T,s,t}$ differs from $\G^{*}_{T,s,t}$ by having $k$ additional unary check nodes within layer $s$, which have neutral weight. By letting 
\begin{align*}
\Delta_{T,s,t}&=\Erw[\ln Z(\G^*_{T,s,t}(\vec m+\vecone_s, \vec \gamma))]-\Erw[\ln Z(\G^{*,1}_{T,s,t}(\vec m, \vec \gamma))]\quad\text{ and}\\
\Delta'_{T,s,t}&=k^{-1}\Erw[\ln Z(\G^*_{T,s,t}(\vec m, \vec \gamma+\vecone_s))]-k^{-1}\Erw[\ln Z(\G^{*,1}_{T,s,t}(\vec m, \vec \gamma))]\end{align*}
we can write \eqref{eq:interpolation} as
\begin{align}
\label{eq:derivAsDelta}
\frac{\partial}{\partial t}\Erw[\ln Z(\G^*_{T,s,t})]= \beta \Delta_{T,s,t} - k\beta \Delta'_{T,s,t}.
\end{align}

The Nishimori property naturally extends to the case of the interpolation model. By choosing neutral weights in \textbf{I2'} this includes $\G^{*,1}_{T,s,t}$. 
\begin{lemma}
\label{lem:nishimori2}
Lemma \ref{lem:nishimoripinned} remains true if we replace $\G^*$ by either $\G^{*,1}_{T,s,t}$ or $\G^{*}_{T,s,t}$.
\end{lemma}
The proof is analogous to Lemma \ref{lem:nishimori}

\begin{lemma}
\label{lem:DDprime}
Let $(\check\SIGMA,\G^*_{T,s,t})$ be chosen from the \textbf{IP} experiment. Let $\vec a$ be a check node with $\partial \vec a$ chosen from $\vec\nu_s^{\otimes k}$ and $\PSI_{\vec a}$ chosen from 
\begin{align*}
\pr[\PSI =\psi] = \frac{p(\psi)}{\xi}\psi(\SIGMA(\partial \vec a)).
\end{align*}
Moreover, let $\vec b$ be a unary check node with $\partial \vec b$ chosen from $\vec \nu_s$ and assign to it a weight function $\PSI_{\vec b}$ defined by
\begin{align*}
\PSI_{\vec b}(\sigma)=\sum_{\tau\in\Omega^k}\PSI'(\tau)\vecone\{\tau_{\vec i}=\sigma\}\prod_{h\neq \vec i}\vec\mu'_h(\tau_h),
\end{align*}
where the index $\vec i\in[k]$, the weight function $\PSI'\in\Psi$ and $\vec\mu'_1,\ldots,\vec\mu'_k$ are chosen from
\begin{align}
\label{eq:distrPreB}
\pr[\vec i = i,(\vec\mu'_1,\ldots,\vec\mu'_k)\in\cA,\PSI'=\psi]\propto p(\psi)\sum_{\tau\in\Omega^k}\vecone\{\tau_i=\check\SIGMA(\partial\vec b)\}\psi(\tau)\int_{\cA}\prod_{j\neq i}\mu'_j(\tau_j)\dd \pi^{\otimes k}(\mu'_1,\ldots,\mu'_k).
\end{align}
Then 
\begin{align}
\Delta_{T,s,t} =&\Erw[\ln \langle \vec\psi_{\vec a}(\SIGMA(\partial \vec a))\rangle_{\G_{T,s,t}^*}],\label{eq:DeltaExpression1}\\
\Delta'_{T,s,t} =&\Erw[\ln \langle \vec\psi_{\vec b}(\SIGMA(\partial \vec b))\rangle_{\G_{T,s,t}^*}] + o_T(1).\label{eq:DeltaExpression2}
\end{align}
\end{lemma}
\begin{proof}
We begin by showing \eqref{eq:DeltaExpression1}. Note that $\G^{*,1}_{T,s,t}(\vec m, \vec \gamma)$ differs from $\G^{*}_{T,s,t}(\vec m, \vec \gamma)$ by having $k$ additional unary check nodes with neutral weight assigned during layer $s$. Of course the sockets chosen as neighbors of these additional check nodes induce a perturbation on the distributions $\vec \nu_\ell$, $\ell\ge s$. However, $\G^{*,1}_{T,s,t}(\vec m, \vec \gamma)$ and $\G^{*}_{T,s,t}(\vec m+\vecone_s, \vec \gamma)$ can be coupled such that afore-said neighborhood is assigned to the the additional $k$-ary check node $\vec a$, while all remaining choices of neighborhoods and weight functions are chosen from the same distribution within both graphs. Given this coupling, we have
\begin{align*}
\frac{Z(\G^{*}_{T,s,t}(\vec m+\vecone_s, \vec \gamma))}{Z(\G^{*,1}_{T,s,t}(\vec m, \vec \gamma))}=&\frac{\sum_{\sigma\in\Omega^k}\vec\psi_{\vec a}(\sigma(\partial \vec a))\prod_{c\in \vec F\setminus\{\vec a\}}\psi_c(\sigma)}{Z(\G^{*,1}_{T,s,t}(\vec m, \vec \gamma))}=\sum_{\sigma\in\Omega^k}\frac{\psi_{\G^{*,1}_{T,s,t}(\vec m, \vec \gamma)}(\sigma)}{Z(\G^{*,1}_{T,s,t}(\vec m, \vec \gamma))}\psi_{\vec a}(\sigma(\partial \vec a))\\
=&\sum_{\sigma\in\Omega^k}\psi_{\vec a}(\sigma(\partial \vec a))\mu_{\G^{*,1}_{T,s,t}(\vec m, \vec \gamma)}(\sigma) = \langle \vec\psi_{\vec a}(\SIGMA(\partial \vec a))\rangle_{\G^{*,1}_{T,s,t}}.
\end{align*}
Taking the logarithm and integrating immediately gives \eqref{eq:DeltaExpression1}. To show \eqref{eq:DeltaExpression2} we couple $\G^{*}_{T,s,t}(\vec m, \vec \gamma+\vecone_s)$ and $\G^{*,1}_{T,s,t}(\vec m, \vec \gamma)$ in a similar fashion. That is in layer $s$, the neighborhood chosen by $\vec b^{(1)}_1,\ldots,\vec b^{(1)}_k$ of $\G^{*,1}_{T,s,t}(\vec m, \vec \gamma)$ is equally assigned to the $k$ additional unary check nodes of $\G^{*}_{T,s,t}(\vec m, \vec \gamma+\vecone_s)$ that result from the positive entry in $\vecone_s$. These $k$ unary nodes are then independently assigned weight functions chosen from \eqref{eq:distrPreB}. We couple all remaining random variables trivially by copying the choice of $\G^{*,1}_{T,s,t}(\vec m, \vec \gamma)$. With this joint distribution, we can write
\begin{align}
&\frac{Z(\G^{*}_{T,s,t}(\vec m, \vec \gamma+\vecone_s))}{Z(\G^{*,1}_{T,s,t}(\vec m, \vec \gamma))}\nonumber\\
=&\frac{\sum_{\sigma\in\Omega^k}\prod_{i=1}^k\vec\psi_{\vec b^{(1)}_i}(\sigma(\partial \vec b^{(1)}_i))\prod_{c\in \vec F\setminus\{\cup_{i=1}^k\vec b^{(1)}_i\}}\psi_c(\sigma)}{Z(\G^{*,1}_{T,s,t}(\vec m, \vec \gamma))}=\left\langle \prod_{i=1}^k \vec\psi_{\vec b^{(1)}_i}(\SIGMA(\partial \vec b^{(1)}_i))\right\rangle_{\G^{*,1}_{T,s,t}}.\label{eq:DeltaSingles}
\end{align}
By taking the logarithm and integrating \eqref{eq:DeltaSingles} gives
\begin{align*}
k\Delta'_{T,s,t}=\Erw[\ln Z(\G^*_{T,s,t}(\vec m, \vec \gamma+\vecone_s))]-\Erw[\ln Z(\G^{*,1}_{T,s,t}(\vec m, \vec \gamma))]=\Erw\left[\left\langle \prod_{i=1}^k \vec\psi_{\vec b^{(1)}_i}(\SIGMA(\partial \vec b^{(1)}_i))\right\rangle_{\G^{*,1}_{T,s,t}}\right].
\end{align*}
While the initial Gibbs measure of $\G^{*,1}_{s,t}$ does not necessarily factorize over the marginals of $\partial \vec b^{(1)}_i$, $i=1,\ldots,k$, Lemma \ref{lem:epsSymmetry} guarantees that our pinned measure $\G^{*,1}_{T,s,t}$ is in fact $(o_T(1),k)$-symmetric with probability $1-o_T(1)$. Because the random set $\vec Y=\{\cup_{i=1}^k \vec b^{(1)}_i\}$ is contiguous with respect to a uniformly random choice $\vec Y'$ from $[n]^k$ this implies
\begin{align*}
\Erw\left[\left\langle \prod_{i=1}^k \vec\psi_{\vec b^{(1)}_i}(\SIGMA(\partial \vec b^{(1)}_i))\right\rangle_{\G^{*,1}_{T,s,t}}\right]=k\Erw\left[\langle \vec\psi_{\vec b^{(1)}_1}(\SIGMA(\partial \vec b^{(1)}_1))\rangle_{\G^{*,1}_{T,s,t}}\right]+o_T(1)
\end{align*}
and thus proves the assertion.
\end{proof}

Before we can verify Proposition \ref{prop:derivIsXi} we have to write out \eqref{eq:DeltaExpression1} and \eqref{eq:DeltaExpression2}. This is a simple, but technical computation.

\begin{claim}
\label{claim:calculations}
With the assumptions of Proposition \ref{prop:derivIsXi} we have
\begin{align}
\Delta_{T,s,t} = & - \frac{1-\xi}{\xi}+\frac{1}{\xi}\sum_{l\ge 2}\frac{1}{l(l-1)}\Erw\left[\prod_{i=1}^k \vec\nu_s(\vec y_i)\left\langle\prod_{h=1}^l 1-\vec\psi(\SIGMA_h(\vec y_1),\ldots,\SIGMA_h(\vec y_k))\right\rangle_{\G^*_{T,s,t}} \right]\label{eq:ADelta}\\
\Delta'_{T,s,t}= & o_T(1)-\frac{1-\xi}{\xi}\nonumber\\
&+\sum_{l\ge 2}\frac{1}{l(l-1)\xi}\Erw\left[\vec\nu_s(\vec y)\left\langle \prod_{h=1}^l1-\sum_{\tau\in\Omega^k}\PSI(\tau)\vecone\{\tau_{\vec i}=\check\SIGMA(\vec y)\}\prod_{j\neq i}\vec\mu_j(\tau_j)\right\rangle_{\G^*_{T,s,t}}\right].\label{eq:ADeltaP}
\end{align}
The definition $\Delta''_{s,t}=\frac{1}{\beta(k-1)}\frac{\partial}{\partial t}\Gamma_{s,t}$ gives
\begin{align}
\Delta''_{s,t}=-\frac{1-\xi}{\xi}+\frac{1}{\xi}\sum_{l\ge 2}\frac{1}{l(l-1)}\Erw\left[\left(1-\sum_{\tau\in\Omega^k}\PSI(\tau)\prod_{j=1}^k\vec\mu_j(\tau_j)\right)^l\right].\label{eq:ADeltaPP}
\end{align} 
\end{claim}

\begin{proof}
We begin by showing \eqref{eq:ADelta}. In the interpolation model $(\check\SIGMA,\G^*_{T,s,t})$, each $k$-ary check node $\vec a$ chooses its neighborhood from $\pr[\partial \vec a = (x_1,\ldots,x_k)]=\vec\nu_s^{\otimes k}(x_1,\ldots,x_k)$ and then obtains a weight function from $\pr[\PSI_{\vec a}|\check\SIGMA=\sigma]=\xi^{-1}p(\psi)\psi\sigma(\partial \vec a)$. Because all weight functions take values in $(0,2)$ writing $\SIGMA_1,\SIGMA_2,\ldots$ for independent samples from $\mu_{\G_{T,s,t}^{*,1}}$ and expanding the logarithm
\begin{align*}
\ln\langle \vec\psi(\SIGMA(y_1),\ldots,\SIGMA(y_k))\rangle_{\G_{T,s,t}^{*,1}}&=-\sum_{l\ge 1}\frac{1}{l}\langle 1-\vec\psi(\SIGMA(y_1),\ldots,\SIGMA(y_k))\rangle^l_{\G_{T,s,t}^{*,1}}\\
&=-\sum_{l\ge 1}\frac{1}{l}\left\langle \prod_{h=1}^l1-\vec\psi(\SIGMA_h(y_1),\ldots,\SIGMA_h(y_k))\right\rangle_{\G_{T,s,t}^{*,1}},
\end{align*}
we obtain
\begin{align}
&\Erw\left[\ln \left\langle \vec\psi_{\vec a}(\SIGMA(\partial \vec a))\right\rangle_{\G_{T,s,t}^{*,1}}\right]\nonumber\\
=&-\sum_{l\ge 1}\frac{1}{l\xi n^k}\sum_{y_1,\ldots,y_k}\Erw\left[\left(\prod_{i=1}^k\vec\nu_s(y_i)\right)\vec\psi(\check\SIGMA(y_1),\ldots,\check\SIGMA(y_k)))\left\langle \prod_{h=1}^l1-\vec\psi(\SIGMA_h(y_1),\ldots,\SIGMA_h(y_k))\right\rangle_{\G_{T,s,t}^{*,1}}\right]\nonumber\\
=&\sum_{l\ge 1}\frac{1}{l\xi n^k}\sum_{y_1,\ldots,y_k}\Erw\left[\left(\prod_{i=1}^k\vec\nu_s(y_i)\right)(1-\vec\psi(\check\SIGMA(y_1),\ldots,\check\SIGMA(y_k)))\right.\\
&\qquad\qquad\qquad\qquad\qquad\qquad\qquad\qquad\cdot\left.\left\langle \prod_{h=1}^l1-\vec\psi(\SIGMA_h(y_1),\ldots,\SIGMA_h(y_k))\right\rangle_{\G_{T,s,t}^{*,1}}\right]\nonumber\\
&\qquad\qquad\quad-\sum_{l\ge 1}\frac{1}{l\xi n^k}\sum_{y_1,\ldots,y_k}\Erw\left[\left(\prod_{i=1}^k\vec\nu_s(y_i)\right)\langle 1-\vec\psi(\SIGMA(y_1),\ldots,\SIGMA(y_k))\rangle^l_{\G_{T,s,t}^{*,1}}\right].\label{eq:Dst1}
\end{align}
By Lemma \ref{lem:nishimori2}, the pairs $(\check\SIGMA,\G_{T,s,t}^{*,1})$ and $(\SIGMA_h,\G_{T,s,t}^{*,1})$, $h\ge 1$ are identically distributed. Thus, we can write
\begin{align*}
&\Erw\left[\left(\prod_{i=1}^k\vec\nu_s(y_i)\right)(1-\vec\psi(\check\SIGMA(y_1),\ldots,\check\SIGMA(y_k)))\left\langle \prod_{h=1}^l1-\vec\psi(\SIGMA_h(y_1),\ldots,\SIGMA_h(y_k))\right\rangle_{\G_{T,s,t}^{*,1}}\right]\nonumber\\
=&\Erw\left[\left(\prod_{i=1}^k\vec\nu_s(y_i)\right)\left\langle \prod_{h=1}^{l+1}1-\vec\psi(\SIGMA_h(y_1),\ldots,\SIGMA_h(y_k))\right\rangle_{\G_{T,s,t}^{*,1}}\right].
\end{align*}
Moreover, \textbf{SYM} gives
\begin{align*}
\frac{1}{\xi n^k}\sum_{y_1,\ldots, y_k}\Erw\left[\left(\prod_{i=1}^k\vec\nu_s(y_i)\right)\langle1-\vec\psi(\SIGMA(y_1),\ldots,\SIGMA(y_k))\rangle_{\G_{T,s,t}^{*,1}}\right]=\frac{1-\xi}{\xi}
\end{align*}
and thus \eqref{eq:Dst1} simplifies to
\begin{align*}
-\frac{1-\xi}{\xi}+\sum_{l\ge 2}\sum_{y_1,\ldots, y_k}\frac{1}{l(l-1)\xi n^k}\Erw\left[\left(\prod_{i=1}^k\vec\nu_s(y_i)\right)\left\langle \prod_{h=1}^{l+1}1-\vec\psi(\SIGMA_h(y_1),\ldots,\SIGMA_h(y_k))\right\rangle_{\G_{T,s,t}^{*,1}}\right].
\end{align*}
Lemma \ref{lem:DDprime} then implies \eqref{eq:ADeltaPP}.

To calculate \eqref{eq:ADeltaP}, recall that in layer $s$ a unary check node $\vec b$ is added by chosing a neighbor $\vec y$ from $\nu_s$ and equipping $\vec b$ with a weight function chosen from \eqref{eq:distrPreB}. For $y\in[n]$ write $\vec b^y$ for a check node chosen in the same way but conditioned on the event that $\partial \vec b = y$. Then
\begin{align}
\Erw\left[\ln \langle \vec\psi_{\vec b}(\SIGMA(\partial \vec b))\rangle_{\G_{T,s,t}^{*,1}}\right]=\sum_{y\in [n]}\Erw\left[\vec\nu_s(y)\ln \langle \vec\psi_{\vec b^y}(\SIGMA(y))\rangle_{\G_{T,s,t}^{*,1}}\right].\label{eq:Dst2}
\end{align}
By \textbf{SYM} the normalization in \eqref{eq:distrPreB} is $
\sum_{\tau\in\Omega^k}\sum_{i=1}^k\vecone\{\tau_i=\check\SIGMA(\partial\vec b)\}\Erw[\PSI(\tau)\prod_{j\neq i}\vec\mu'_j(\tau_j)]=k\xi$. By writing out \eqref{eq:distrPreB} in equation \eqref{eq:Dst2} we get
\begin{align}
\frac{1}{nk\xi }\sum_{y\in[n],i\in[k]}&\Erw\left[\vec\nu_s(y)\sum_{\tau\in\Omega^k}\vecone\{\tau_i=\check\SIGMA(x)\}\PSI(\tau)\right.\\
&\qquad\qquad\qquad\cdot\left.\prod_{j\neq i}\vec\mu'_j(\tau_j)\ln \left\langle \sum_{\sigma\in\Omega^k}\vecone\{\sigma_i = \SIGMA(y)\}\PSI(\sigma)\prod_{j\neq i}\vec\mu'_j(\sigma_j)\right\rangle_{\G_{T,s,t}^{*,1}}\right].\label{eq:Dst3}
\end{align}
Again, we utilize the fact that our weight functions take values within $(0,2)$ and expand the logarithm to write \eqref{eq:Dst3} as
\begin{align*}
&-\sum_{y\in[n]}\sum_{i=1}^k\sum_{l\ge 1}\frac{1}{nkl\xi}\Erw\left[\vec\nu_s(y)\sum_{\tau\in\Omega^k}\vecone\{\tau_i=\check\SIGMA(x)\}\PSI(\tau)\prod_{j\neq i}\vec\mu'_j(\tau_j)\right.\\
&\qquad\qquad\qquad\qquad\qquad\qquad\qquad\qquad\left.\cdot\left\langle \prod_{h=1}^l 1-\sum_{\sigma\in\Omega^k}\vecone\{\sigma_i = \SIGMA_h(y)\}\PSI(\sigma)\prod_{j\neq i}\vec\mu'_j(\sigma_j)\right\rangle_{\G_{T,s,t}^{*,1}}\right]\\
=&-\sum_{y\in[n]}\sum_{i=1}^k\sum_{l\ge 1}\frac{1}{nkl\xi}\Erw\left[\vec\nu_s(y)\left(1-\sum_{\tau\in\Omega^k}\vecone\{\tau_i=\check\SIGMA(x)\}\PSI(\tau)\prod_{j\neq i}\vec\mu'_j(\tau_j)\right)\right.\\
&\qquad\qquad\qquad\qquad\qquad\qquad\qquad\qquad\cdot\left\langle\prod_{h=1}^l1-\sum_{\sigma\in\Omega^k}\vecone\{\sigma_i = \SIGMA_h(y)\}\PSI(\sigma)\prod_{j\neq i}\vec\mu'_j(\sigma_j)\right\rangle_{\G_{T,s,t}^{*,1}}\\
&\qquad\qquad\qquad\qquad\qquad\qquad\qquad-\left.\vec\nu_s(y)\left\langle\prod_{h=1}^l1-\sum_{\sigma\in\Omega^k}\vecone\{\sigma_i = \SIGMA_h(y)\}\PSI(\sigma)\prod_{j\neq i}\vec\mu'_j(\sigma_j)\right\rangle_{\G_{T,s,t}^{*,1}}\right]
\end{align*}
Once again we simplify the expression by using that the distributions $\check\SIGMA$ and $\SIGMA_h$, $h\ge 1$ coincide 
\begin{align*}
&\Erw\left[\vec\nu_s(y)\left(1-\sum_{\tau\in\Omega^k}\vecone\{\tau_i=\check\SIGMA(x)\}\PSI(\tau)\prod_{j\neq i}\vec\mu'_j(\tau_j)\right)\right.\\
&\qquad\qquad\qquad\qquad\left.\cdot\left\langle\prod_{h=1}^l1-\sum_{\sigma\in\Omega^k}\vecone\{\sigma_i = \SIGMA_h(y)\}\PSI(\sigma)\prod_{j\neq i}\vec\mu'_j(\sigma_j)\right\rangle_{\G_{T,s,t}^{*,1}}\right]\\
=&\Erw\left[\vec\nu_s(y)\left\langle\prod_{h=1}^{l+1}1-\sum_{\sigma\in\Omega^k}\vecone\{\sigma_i = \SIGMA_h(y)\}\PSI(\sigma)\prod_{j\neq i}\vec\mu'_j(\sigma_j)\right\rangle_{\G_{T,s,t}^{*,1}}\right]
\end{align*}
and obtain
\begin{align*}
&\Erw\left[\ln \langle \vec\psi_{\vec b}(\SIGMA(\partial \vec b))\rangle_{\G_{T,s,t}^{*,1}}\right]\\
=&-\frac{1}{nk\xi}\sum_{y\in[n]}\sum_{i\in[k]}\Erw\left[\vec\nu_s(y)\left\langle1-\sum_{\sigma\in\Omega^k}\vecone\{\sigma_i = \SIGMA_1(y)\}\PSI(\sigma)\prod_{j\neq i}\vec\mu'_j(\sigma_j)\right\rangle_{\G_{T,s,t}^{*,1}}\right]\\
&+\sum_{l\ge 2}\frac{1}{nk\xi l(l-1)}\sum_{y\in[n]}\sum_{i\in[k]}\Erw\left[\vec\nu_s(y)\left\langle\prod_{h=1}^l1-\sum_{\sigma\in\Omega^k}\vecone\{\sigma_i = \SIGMA_1(y)\}\PSI(\sigma)\prod_{j\neq i}\vec\mu'_j(\sigma_j)\right\rangle_{\G_{T,s,t}^{*,1}}\right]\\
=&-\frac{1-\xi}{\xi}\\
&+\sum_{l\ge 2}\frac{1}{nk\xi l(l-1)}\sum_{y\in[n]}\sum_{i\in[k]}\Erw\left[\vec\nu_s(y)\left\langle\prod_{h=1}^l1-\sum_{\sigma\in\Omega^k}\vecone\{\sigma_i = \SIGMA_1(y)\}\PSI(\sigma)\prod_{j\neq i}\vec\mu'_j(\sigma_j)\right\rangle_{\G_{T,s,t}^{*,1}}\right]
\end{align*}
by employing \textbf{SYM} in the final equation.

Finally let us derive \eqref{eq:ADeltaPP}. By definition $\Gamma_{s,t}=\xi^{-1}(s+t-1)\beta(k-1)\Erw[\Lambda(\sum_{\tau\in\Omega^k}\PSI(\tau)\prod_{j=1}^k\vec\mu^{(\pi)}_j(\tau_j))]$, where we write $\vec\mu^{(\pi)}_1,\vec\mu^{(\pi)}_2,\ldots$ for independent samples from $\pi$. As in the previous cases, we perform the same procedure of expanding the logarithm and simplifying the telescopic sum to obtain
\begin{align*}
&\frac{1}{\beta(k-1)}\frac{\partial}{\partial t}\Gamma_{s,t}=\xi^{-1}\Erw\left[\Lambda\left(\sum_{\tau\in\Omega^k}\PSI(\tau)\prod_{j=1}^k\vec\mu^{(\pi)}_j(\tau_j)\right)\right]\\
=&\xi^{-1}\sum_{l\ge1}\frac{1}{l}\Erw\left[\left(1-\sum_{\tau\in\Omega^k}\PSI(\tau)\prod_{j=1}^k\vec\mu^{(\pi)}_j(\tau_j)\right)\left(1-\sum_{\tau\in\Omega^k}\PSI(\tau)\prod_{j=1}^k\vec\mu^{(\pi)}_j(\tau_j)\right)^l\right.\\
&\left.\qquad\qquad\qquad\qquad\qquad\qquad\qquad\qquad\qquad-\left(1-\sum_{\tau\in\Omega^k}\PSI(\tau)\prod_{j=1}^k\vec\mu^{(\pi)}_j(\tau_j)\right)^l\right]\\
=&-\xi^{-1}\Erw\left[1-\sum_{\tau\in\Omega^k}\PSI(\tau)\prod_{j=1}^k\vec\mu^{(\pi)}_j(\tau_j)\right]+\xi^{-1}\sum_{l\ge 2}\frac{1}{l(l-1)}\Erw\left[\left(1-\sum_{\tau\in\Omega^k}\PSI(\tau)\prod_{j=1}^k\vec\mu^{(\pi)}_j(\tau_j)\right)^l\right]\\
=&-\frac{1-\xi}{\xi}+\frac{1}{\xi}\sum_{l\ge 2}\frac{1}{l(l-1)}\Erw\left[\left(1-\sum_{\tau\in\Omega^k}\PSI(\tau)\prod_{j=1}^k\vec\mu^{(\pi)}_j(\tau_j)\right)^l\right].
\end{align*}
\end{proof}

\begin{proof}[Proof of Proposition \ref{prop:derivIsXi}]
The assertion is now immediate as \eqref{eq:derivAsDelta} and Claim \ref{claim:calculations} yield
\begin{align*}
\frac{\partial}{\partial t}\Phi_{T,s}(t)=n^{-1}\left[o_T(1)+\beta\Delta_{s,t}-k\beta\Delta'_{s,t}+\beta(k-1)\Delta''_{s,t}\right]=o_T(1)n^{-1}+\frac{\beta}{n}\sum_{l\ge 2}\frac{\Xi_{s,t,l}}{l(l-1)}.
\end{align*}
\end{proof}

We complete the proof of Proposition \ref{prop:derivBound} by comparing the expressions $\Xi_{s,t,l}$ to a non-negative value given by \textbf{POS}. To this end let $\vec\rho_1,\vec\rho_2,\ldots$ be a sequence of independently samples from $\pi_{\G^*_{T,s,t}}$, let $\vec\mu_1,\vec\mu_2,\ldots$ be independently chosen from $\pi$. Define
\begin{align}
\Xi'_{s,t,l}=&\Erw\left[\left(1-\sum_{\sigma\in\Omega^k}\PSI(\sigma)\prod_{i=1}^k\vec\rho_i(\tau_i)\right)^l+(k-1)\left(1-\sum_{\tau\in\Omega^k}\vec\psi(\tau)\prod_{i=1}^k\vec\mu_i(\tau_i)\right)^l\right.\nonumber\\
&\left.\qquad-\sum_{i=1}^k\left(1-\sum_{\tau\in\Omega^k}\vec\psi(\tau)\vec\rho_1(\tau_1)\prod_{j\neq i}\vec\mu_j(\tau_j)\right)^l\right].\label{eq:XiPrime}
\end{align}
Let $\SIGMA$ be a chosen from $\mu_{\G^*_{T,s,t}}$. By Lemma \ref{lem:nishimori2} the pairs $(\SIGMA^*,\G^*_{T,s,t})$ and $(\SIGMA,\G^*_{T,s,t})$ are identically distributed. Thus, for any choice of $s,t$ the mean of the empirical marginal distribution $\pi_{\G^*_{T,s,t}}$ is given by the uniform distribution. As \textbf{POS} holds, expanding $\Lambda$ shows that the expression \eqref{eq:XiPrime} is non-negative for any $l\ge0$. 
\begin{proof}[Proof of Proposition \ref{prop:derivBound}]
With this consideration the assertion follows from Proposition \ref{prop:derivIsXi} by verifying
\begin{align}
\label{eq:XiAreClose}
|\Xi_{s,t,l}-\Xi'_{s,t,l}|=o_T(1)
\end{align}
for any choice of $s,t$ and $l$.

For the first and third summand in \eqref{eq:XiPrime} we will use the Pinning Lemma to sufficiently regularize the underlying graph models, such that exchangeability with the corresponding term from $\Xi_{s,t,l}$ is possible. For the time being, consider the first suammnd. Observe that for independently uniform choices $\vec y_1,\vec y_2,\ldots$ among the variable nodes $[n]$, we can write
\begin{align*}
&\Erw\left[\left(1-\sum_{\sigma\in\Omega^k}\PSI(\sigma)\prod_{i=1}^k\vec\rho_i(\sigma_i)\right)^l\bigg|\G^*_{T,s,t}\right]\\
=&\Erw\left[\prod_{i=1}^k\vec\nu_s(\vec y_i)\left(1-\sum_{\sigma\in\Omega^k}\PSI(\sigma)\prod_{i=1}^k\mu_{\G^*_{T,s,t}}(\sigma_i)\right)^l\bigg|\G^*_{T,s,t}\right].
\end{align*}
Hence, for any $\psi\in\Psi$ the triangle inequality and the Cauchy-Schwarz inequality yield
\begin{align}
&\left|\frac{1}{n^k}\sum_{y_1,\ldots,y_k}\prod_{i=1}^k\vec\nu_s(y_i)\langle 1-\psi(\SIGMA(y_1),\ldots,\SIGMA(y_k)\rangle^l_{\G^*_{T,s,t}}-\Erw\left[\left(1-\sum_{\sigma\in\Omega^k}\psi(\sigma)\prod_{i=1}^k\vec\rho_i(\sigma_i)\right)^l\bigg|\G^*_{T,s,t}\right]\right|\nonumber\\
\le&\frac{1}{n^k}\sum_{y_1,\ldots,y_k}\prod_{i=1}^k\vec\nu_s(y_i)\left|\vphantom{\Erw\left[\left(1-\sum_{\sigma\in\Omega^k}\psi(\sigma)\prod_{i=1}^k\mu_{\G^*_{T,s,t}}(\sigma_i)\right)^l|\G^*_{T,s,t}\right]}\langle 1-\psi(\SIGMA(y_1),\ldots,\SIGMA(y_k)\rangle^l_{\G^*_{T,s,t}}\right.\nonumber\\
&\left.\qquad\qquad\qquad\qquad\qquad\qquad-\Erw\left[\left(1-\sum_{\sigma\in\Omega^k}\PSI(\sigma)\prod_{i=1}^k\mu_{\G^*_{T,s,t}}(\sigma_i)\right)^l\bigg|\G^*_{T,s,t}\right]\right|\nonumber\\
\le&n^{-k}\left(\sum_{y_1,\ldots,y_k}\prod_{i=1}^k\vec\nu_s(y_i)^2\right)^{1/2}\label{eq:degreenorm}\\
&\cdot\left(\sum_{y_1,\ldots,y_k}\left|\langle 1-\psi(\SIGMA(y_1),\ldots,\SIGMA(y_k)\rangle^l_{\G^*_{T,s,t}}-\Erw\left[\left(1-\sum_{\sigma\in\Omega^k}\psi(\sigma)\prod_{i=1}^k\mu_{\G^*_{T,s,t}}(\sigma_i)\right)^l\bigg|\G^*_{T,s,t}\right]\right|^2\right)^{1/2}.\nonumber
\end{align}
We can rewrite \eqref{eq:degreenorm} as the ratio between the 2- and 1-norm of the vector $(\vec \delta_s(v))_{v\in n}$. The expression is maximized by sparse vectors, where each vector contains at most $\sum_v\vec\delta_1(v)/(\alpha \max_v \vec\delta_1(v))$ non-zero entries. Hence, we can bound the factor \eqref{eq:degreenorm} by
\begin{align*}
n^{-k}\left(\sum_{y_1,\ldots,y_k}\cdot\prod_{i=1}^k\vec \nu_s(y_i)^2\right)^{1/2}\le n^{-k}\left(\frac{\sqrt{\sum_{v}\vec\delta_s(v)^2}}{\sum_v\vec\delta_s(v)}\right)^k\le\frac{O(1)}{n^{k/2}}.
\end{align*}
Moreover, as each $\psi$ evaluates in $(0,2)$, Lemma \ref{lem:nishimori2} implies that for any $C>0, \eps>0, l\ge1$ there is $\delta>0$ such that if $\G^*_{T,s,t}$ is $\delta$-symmetric, we have
\begin{align*}
C n^{-k}\sum_{y_1,\ldots,y_k}\left|\langle 1-\psi(\SIGMA(y_1),\ldots,\SIGMA(y_k)\rangle^l_{\G^*_{T,s,t}}-\Erw\left[\left(1-\sum_{\sigma\in\Omega^k}\psi(\sigma)\prod_{i=1}^k\mu_{\G^*_{T,s,t}}(\sigma_i)\right)^l\bigg|\G^*_{T,s,t}\right]\right|<\eps.
\end{align*}
In the same way, for any $\psi\in\Psi$ and $i\in [k]$ we can bound
\begin{align*}
&\left|\frac{1}{n}\sum_{y}\prod_{i=1}^k\vec\nu_s(y_i)\langle 1-\sum_{\tau\in\Omega^k}\psi(\tau)\vecone\{\tau_i=\SIGMA_i(y)\}\prod_{j\neq i}\vec\mu_j(\tau_j)\rangle^l_{\G^*_{T,s,t}}\right.\\
&\left.\qquad\qquad\qquad\qquad\qquad\qquad-\Erw\left[\left(1-\sum_{\tau\in\Omega^k}\psi(\tau)\vec\rho_1(\tau_i)\prod_{j\neq i}\vec\mu_j(\tau_j)\right)^l\bigg|\G^*_{T,s,t}\right]\right|<\eps.
\end{align*}
By Lemma \ref{lem:PinningLemma} we can guarantee that $\G^*_{T,s,t}$ is $o_T(1)$-symmetric with probability $1-o_T(1)$, which implies \eqref{eq:XiAreClose} with probability $1-o_T(1)$.
\end{proof}

\subsection*{Acknowledgement} We thank Amin Coja-Oghlan for helpful discussions.

\end{document}